\documentclass[11pt]{article}
\usepackage[utf8]{inputenc}

\usepackage{fullpage}

\usepackage{algorithm}
\usepackage{algorithmicx}
\usepackage[noend]{algpseudocode}

\usepackage{amsmath}
\usepackage{amsfonts}
\usepackage{amsthm}
\usepackage{bbm}
\usepackage{subcaption}
\usepackage{graphicx}
\usepackage{caption}

\usepackage[numbers]{natbib}

\usepackage{comment}

\usepackage{todonotes}
\usepackage{thmtools}

\usepackage{soul}
\usepackage{complexity}

\usepackage{mdframed}

\usepackage{hyperref}
\hypersetup{
    colorlinks=true,
    linkcolor=blue,
    filecolor=magenta,      
    urlcolor=blue,
    pdfpagemode=FullScreen,
    }

\DeclareMathOperator*{\argmax}{arg\,max}

\newcommand{\probP}{\text{I\kern-0.15em P}}

\newcommand{\emax}{E_{\text{max}}}
\newcommand{\dks}{\text{D}k\text{S}}
\newcommand{\ndks}{\text{N-D}k\text{S}}
\newcommand{\dutting}{D\"utting et al. }
\newcommand{\est}{\widehat{\operatorname{deg}}}
\newcommand{\ic}{\textit{IC}}

\newcommand{\kcd}{\textit{k}-CD}
\newcommand{\ugsc}{U-GSC}
\newcommand{\lsac}{LSAC}

\newtheorem{theorem}{Theorem}[section]
\newtheorem{lemma}[theorem]{Lemma}

\newtheorem{proposition}[theorem]{Proposition}
\newtheorem{observation}[theorem]{Observation}

\theoremstyle{definition}
\newtheorem{definition}[theorem]{Definition}

\let\vec\mathbf

\title{On Supermodular Contracts and Dense Subgraphs}
\author{
   \normalsize{Ramiro Deo-Campo Vuong} \thanks{Thomas Lord Department of Computer Science, University of Southern California. } \\
    \small{rdeocamp@usc.edu}
   \and
   \normalsize{Shaddin Dughmi} \footnotemark[1] \\
    \small{shaddin@usc.edu}
   \and
   \normalsize{Neel Patel} \footnotemark[1] \\
    \small{neelbpat@usc.edu}
   \and
   \normalsize{Aditya Prasad} \footnotemark[1] \\
    \small{aprasad4@usc.edu}
}

\date{\vspace{-5ex}}

\begin{document}

\maketitle
\begin{abstract}
    We study the combinatorial contract design problem, introduced and studied by \dutting (2021, 2022), in both the single and multi-agent settings. Prior work has examined the problem when the principal's utility function is submodular in the actions chosen by the agent(s).
    We complement this emerging literature with an examination of the problem when the principal's utility is supermodular.
    Our results apply to the unconstrained contract design problem in the binary outcome case (i.e., the principal's task succeeds or fails), and to the linear contract design problem more generally.

    In the single-agent setting, we obtain a strongly polynomial time algorithm for the optimal contract.
    This stands in contrast to the NP-hardness of the problem with submodular principal utility due to \dutting (2021).
    This result has two technical components, the first of which applies beyond supermodular or submodular utilities.
    First, we describe a simple divide-and-conquer algorithm which enumerates all the ``breakpoints'' of the principal's utility function in strongly polynomial time.
    This result strengthens and simplifies analogous enumeration algorithms from \dutting (2021), and applies to any nondecreasing valuation function for the principal.
    Second, we show that supermodular valuations lead to a polynomial number of breakpoints, analogous to a similar result by \dutting (2021) for gross substitutes valuations.

    In the multi-agent setting, we obtain a mixed bag of positive and negative results.
    First, we show that it is NP-hard to obtain any finite multiplicative approximation, or an additive FPTAS.
    This stands in contrast to the submodular case, where efficient computation of approximately optimal contracts was shown by \dutting (2022).
    Second, we derive an additive PTAS for the problem in the instructive special case of graph-based supermodular valuations, and equal costs.
    En-route to this result, we discover an intimate connection between the multi-agent contract problem and the notorious k-densest subgraph problem.
    We build on and combine techniques from the literature on dense subgraph problems to obtain our additive PTAS.
    We leave open the intriguing, and seemingly quite challenging, question of whether an additive PTAS exists more generally for multi-agent supermodular contracts.
\end{abstract}

\section{Introduction}
Contract theory is one of the backbones of microeconomics, and plays a central role in the markets of services (cf., the 2016 Nobel Prize in Economics for Hart and Holmström \cite{nobelprice2016}).
In the past few years, contract theory has gained growing interest in the CS-Econ community due to its natural connections to optimization theory and wide ranging applications to real-life economic markets including crowdsourcing platforms \cite{ho2014adaptive}, social media influencing, blockchain-based smart contracts \cite{cong2019blockchain}, and incentivizing quality healthcare \cite{bastani2016analysis}. 

In the classic hidden action principal-agent contract model \cite{holmstrom1979moral,grossman1992analysis}, the principal seeks to incentivize the agent to perform a costly action which determines the distribution of the possible rewards to the principal.
The principal incentivizes the agent via monetary transfer or ``contract" based on the observable reward outcome and not on the (hidden) action taken by the agent. 
Given a contract, the agent performs the action which yields her maximum expected utility.
The principal's goal is to maximize her expected reward minus the transfer paid to the agent.

The recent line of computational work on combinatorial contracts \cite{dutting2022combinatorial,duetting2022multi} considers the hidden-action principal-agent model in a natural combinatorial setting that captures the complex dependencies between the agent's actions.
Combinatorial contracts play a central role in the combinatorial markets of services just as combinatorial auctions are crucial in combinatorial markets of goods. 

In the single-agent combinatorial contract model, introduced by \dutting \cite{dutting2022combinatorial}, the principal seeks to delegate a task to the agent, resulting in one of two \emph{outcomes}: success (the principal gets the reward of $1$) or failure (the principal gets no reward).
There is a set of underlying costly actions $[n] = \{1,\dots, n\}$. 
Any subset of actions $S\subseteq [n]$ induces a distribution over outcomes. 
Selecting set $S$ yields expected reward $f(S)$ where $f(\cdot )$ is a \emph{reward} function defined over subsets of actions $2^{[n]}$. 
In other words, if the agent performs a set of actions $S\subseteq [n]$ then the probability of the task being successful is $f(S)$ and the agent incurs a total cost of $\sum_{i\in S} c_i$, where $c_i$ is the cost of action $i\in [n]$.

Follow up work \cite{duetting2022multi} considers the hidden action model in the multi-agent setting.
In this setting, the principal interacts with a set of agents $[n]$. 
Agents have a binary action: agent $i\in [n]$ can exert effort and incur cost $c_i$ or not. 
The action/effort taken by the set of agents $S\subseteq [n]$ induces a distribution over the possible outcomes and the principal obtains a reward of $f(S)$. 
In both \cite{dutting2022combinatorial,duetting2022multi}, the reward function exhibits diminishing returns of an action, i.e., it is \emph{submodular}. 

Prior works in algorithmic contract theory have not addressed settings with complements between actions --- an action has added value when other actions are taken.
In several real-life instances, different actions exhibit natural complementary dependencies.
For example, a healthcare startup (the principal) delegates the training of their prediction model to a consultancy firm (the agent) due to a lack of computational resources.
In this case, the data analysis team of the consultancy firm will more effectively train the model when the data mining team exerts effort to clean up the data.
Consider another example where a firm (the principal) delegates the construction of their office to a construction company.
For construction to be successful, it must meet a given a set of requirements (i.e. building standards code). 
Here, the contribution of the architecture team is negligible towards completion of the office without the efforts of an engineering team and vice-versa.

In this work, we contribute to the emerging understanding of hidden-action combinatorial contract problems in both single and multi-agent settings when the reward function is \emph{supermodular}.
The supermodularity of the reward function captures natural complementarities between the actions.

We observe that in both the single agent and multi-agent settings, the results for supermodular and submodular reward functions stand in contrast to each other. 
In the single-agent setting, calculating the optimal contract for general submodular reward functions is shown to be NP-Hard \cite{dutting2022combinatorial}. 
In contrast, we provide a strongly polynomial time algorithm that computes an optimal contract when the reward function is supermodular.
In the multi-agent setting, this dynamic is inverted: the multi-agent setting with any submodular reward function admits a constant approximation due to \cite{duetting2022multi}.
However, in this work, we show computational hardness of the multi-agent setting with supermodular reward functions, ruling out any bounded multiplicative approximations and an additive FPTAS. 
Then we design an additive PTAS for a special class of multi-agent contracts with graph-supermodular reward functions, leaving open the existence of such a PTAS more generally. 

\subsection*{Our Results and Techniques}
\subsubsection*{Single Agent Setting} 
When the outcome is binary and the principal sets contract $t$, a rational agent takes the set of actions $S \subseteq [n]$ maximizing $t\cdot f(S) - \sum_{i \in S}c_i$.
Thus, the optimal solution to the agent's problem is the result of a \textit{demand query} for $f$ with price $\frac{c_i}{t}$ for each action $i\in [n]$.
\dutting  \cite{dutting2022combinatorial} show that the optimal contract for the single agent setting is one of the \emph{break points} of the agent's utility function.\footnote{
    The principal's and agent's utility functions have break points at the same contracts. We note that \cite{dutting2022combinatorial} refers to break points as critical values.
}
First, we obtain an algorithm that enumerates all break points with $2 |\mathcal{D}_{f, c}| + 1$ calls to the demand query oracle, where $\mathcal{D}_{f, c}$ is the set of break points.
This yields a strongly polynomial time algorithm in the demand oracle model when the number of break points is polynomial.

\vspace{.2 cm}
\noindent 
\textbf{Main Result 1 (Informal) (Theorem \ref{thm:single_agent_algo}):} If the size of the of the break point set $\mathcal{D}_{f, c}$ is polynomial in the number of actions, then there exists a strongly polynomial time algorithm for computing an optimal contract in the demand oracle model.

\vspace{.2 cm}

This result strengthens prior work by \dutting \cite{dutting2022combinatorial}.
They present a binary-search inspired weakly polynomial time algorithm for combinatorial contract instances with a polynomial sized $\mathcal{D}_{f, c}$. 
In addition, our general algorithm also recovers and simplifies their strongly polynomial time algorithm for gross substitute reward functions.

Instead of binary search, our algorithm relies fundamentally on \textit{intersection contracts}: the contract at which two set of actions offer the same utility to the agent.
We observe that each break point is an intersection contract of some pair of demanded sets, and give a recursive algorithm to enumerate all breakpoints. 
Our algorithm finds the set in demand at the intersecting contract of the two ``boundary'' demanded sets. 
Either the intersection contract is a break point or a new demanded set is found upon which our algorithm recurses.

Our second positive result is the application of our general algorithm to the case of supermodular reward functions: we show that the size of the break point set is linear in the number of actions.
Since supermodular functions admit a strongly polynomial time demand oracle \cite{iwata2001polysupermodular}, our general result implies a strongly polynomial time algorithm for single-agent supermodular contracts in the value oracle model.
This stands in contrast to the proven computational hardness of settings with submodular reward \cite{dutting2022combinatorial}.

\subsubsection*{Multi Agent Setting} 
Recently, the multi-agent hidden-action setting has received more attention \cite{duetting2022multi, castiglioni2023multiagent} and is known to be significantly more challenging than the single-agent setting.
\dutting \cite{duetting2022multi} show that calculating an optimal contract in even the simplest class of reward functions --- additive rewards --- is NP-hard.
However, they obtain a constant multiplicative approximation of an optimal contract when the reward function is submodular and, more generally, fractionally subadditive.

We contribute to the emerging understanding of the multi-agent setting by showing that supermodular reward settings are more difficult than their submodular counterparts --- even graph supermodular reward functions admit neither an efficient multiplicative approximation algorithm nor an additive \emph{fully polynomial time approximation scheme} (FPTAS).
This result contrasts with the constant multiplicative approximation algorithms for submodular contracts from~\cite{duetting2022multi}.

\vspace{.2 cm}
\noindent
\textbf{Main Result 2 (Theorem \ref{thm:multi_agent_supermodular_hardness}):} The supermodular multi-agent contract problem admits no polynomial time constant multiplicative approximation algorithm nor an additive FPTAS unless $P=NP$. The hardness holds even for uniform costs and graph supermodular rewards, denoted as \ugsc.

\vspace{.2 cm}
First, we note that the principal's problem of finding an optimal contract reduces to a combinatorial optimization problem of finding the set of nodes $S\subseteq V$ of a given undirected graph $G=(V,E)$ that maximizes
\begin{equation*}
    \mu_p(S) = \left(1 - \sum_{i \in S}\frac{c_i
    \cdot \emax}{\deg_S(i)} \right) \left(\frac{|E(S)|}{\emax} \right)   , 
\end{equation*}
where $\emax = \binom{n}{2}$. 
A set yielding high utility will have many edges and few nodes with low degree.
A clique is a natural extreme of such a set, and in fact, we show hardness via a log-space reduction from the decision version of the $k$-Clique problem.
In our reduction, if the underlying graph contains a $k$-clique, the principal obtains utility at least $\Omega(\frac{1}{n^2})$. 
Otherwise, the objective function is always non-positive.
Our reduction rules out any bounded multiplicative approximation scheme and an additive FPTAS.

We observe that a similar log-space reduction converts an additive PTAS for the U-GSC problem to a polynomial time algorithm for a gap version of the \emph{normalized densest $k$ subgraph (\ndks)} problem which we call \emph{linear sized almost clique (\lsac)}.
In general, almost clique (with sublinear $k$) captures the hardness of \ndks \ \cite{braverman2017eth}.
Importantly, it appears our problem is not subsumed by \ndks; we present an example which rules out simple reductions from U-GSC to \ndks.
We also note that the known techniques for \ndks \ \cite{arora1995polynomial, barman2015approximating} do not appear to transfer to our problem since our objective cannot be represented as a polynomial integer program.

\vspace{.2 cm}
\noindent
\textbf{Main Result 3 (Theorem \ref{thm:PTAS}) Informal:}
    There exists an additive PTAS for uniform cost graph supermodular contracts in the multi-agent setting. 

\vspace{.2 cm}

Motivated by the connection between \ugsc \ and \ndks, we revisit established techniques for approximating dense subgraphs.
Our algorithm uses the general framework in \ndks \ algorithms by Barman and Arora et al.~\cite{barman2015approximating, arora1995polynomial}.
Their \ndks \ algorithms sample nodes to estimate the degree of every node to a dense subgraph, use said degree estimates in a convex program to reconstruct a ``near-optimal'' fractional solution, and employ a rounding scheme to recover an integral solution.
\ugsc \ poses an additional difficulty: whereas the density objective is robust to the inclusion a few ``low degree nodes'', our objective plummets even when selecting one such node. Moreover, the techniques to solve \ndks \ in previous works \cite{arora1995polynomial,barman2015approximating,feige1997densest} heavily rely on a natural relaxation of \ndks \ as a quadratic program.
Such a program cannot represent the \ugsc \ objective due to the \textit{sum of inverse degree} term.
To compensate for the sensitivity to low degree nodes we show the following, where $S^*$ denotes an optimal solution to U-GSC:

\begin{enumerate}
    \item There exists a subset $\tilde S \subseteq S^*$ which is ``near-optimal," and for all $v\in \tilde S$, $\deg_{\tilde S}(v) = \Omega(n)$.
    \item There is an efficient sampler that samples $O(\log n)$-nodes from the the set $\tilde S$ without any information about $\tilde S$. 
\end{enumerate}
    
We show that the set satisfying Property 1 is a core of $S^*$.\footnote{Recall that a $k$-core is the unique maximal subgraph in which every node has degree at least $k$.}
To prove that a core has near optimal utility, we analyze the change in a relaxed version of the principal's objective for every iteration of the coring algorithm.\footnote{Notably, relaxing the objective appears necessary, as the unrelaxed objective varies non-monotonically with iterations of the coring algorithm.}
Via this analysis, we show that $\tilde S$ is ``near-optimal'' (w.r.t. original objective), linear sized ($|\tilde{S}| = \Omega(n)$), and only contains nodes of high degree.
The core $\tilde S$ is conducive to reconstruction because it lacks low degree nodes, so we approximate $\tilde S$ instead of $S^*$.

Our efficient sampler relies on the fact that $\tilde S$ is of linear size.
The idea here is that any random vertex sampled from $V$ lies in the set $\tilde S$ with probability $\Omega(1)$.
Hence, when $O(\log n)$ nodes are sampled uniformly from $V$, the entire sample lies in the set $\tilde S$ with probability $\frac{1}{\poly(n)}$.
Independently repeating this sampling process $\poly (n)$ times guarantees that with probability at least $1 - \frac{1}{\poly(n)}$, one of the samples completely lies inside $\tilde S$.
Intuitively, with a sample of $\Omega(\log n)$ nodes from $\tilde{S}$, degree estimates concentrate around their true value.
A similar sampler was used by Daskalakis et al. \cite{daskalakis2011oblivious} to compute an approximate equilibrium in small probability games.

With the above properties, we accurately estimate node degrees with high probability.
Using these estimates in an LP that serves as a proxy to the principal's objective, we calculate an approximate fractional solution.
Finally, we employ randomized rounding to derive a ``near-optimal'' integral solution.
The technicalities of our argument are presented in Section~\ref{sec:PTAS}.

\subsection*{Related Work}

\paragraph{Contract Theory}
Algorithmic contract theory is an emerging frontier in CS-Econ research. 
The exploration of algorithmic aspects of contract theory began in \cite{babaioff2006combinatorial,feldman2007hidden}, which consider a model where the agent can take no more than one of $n$ explicitly-given actions.
These papers construct a linear program for each action to calculate the optimal contract, which in general may be intricate and nonlinear. On the other hand, \dutting \cite{dutting2019simple} shows that the linear contracts --- which are simpler both conceptually and computationally -- are optimal or approximately optimal in a max-min sense, given partial knowledge of reward distributions associated with the various actions. 

Particularly related to this paper is the work of \dutting \cite{dutting2022combinatorial}, which considers algorithmic contract design for a single agent when actions are combinatorial. Specifically, an agent selects a subset of $n$ given actions. They focus on costs that are additive, and principal rewards which exhibit dimishing marginal returns (i.e., are submodular).
They show hardness in general for submodular rewards, and give a polynomial time algorithm for the subclass of gross substitutes rewards.

Also related to this paper is the work on algorithmic contract design for multiple agents. Since the agents make decisions independently, the the problem here is intrinsically combinatorial in nature. We have observed two general models of multi-agent contracts.
The first model is that which is most related to this paper: $n$ agents decide whether or not to exert effort towards a particular task, and the principal only observes the task's outcome.
Babaioff et al.~\cite{babaioff2006combinatorial} and Emek and Feldman~\cite{emek2012computing} specifically study the setting in which each agent has an individual (random) outcome, and a simple boolean function maps the individual outcomes to the result of the task. 
Babaioff et al. give an algorithm to compute an optimal contract for AND networks, and Emek and Feldman show that computing an optimal contract for OR networks --- a special case of submodular functions --- is NP-hard, but admits an FPTAS.
\dutting~\cite{duetting2022multi} generalize the prior work and present a constant factor approximation algorithm for submodular, and more generally XOS, reward functions.

The second model of multi-agent contract design was studied by Castoglioni et al.~\cite{castiglioni2023multiagent}. Here, agents complete individual tasks, the outcome of which is observed by the principal. Since the principal's payment to each agent depends on their own individual outcome, this effectively removes externalities between agents. They show that, under mild regularity conditions, an optimal contract can be computed when rewards are  supermodular, and can be  approximated when rewards are submodular.

\paragraph{Graph Density and Bimatrix Game Nash Equilibrium}
For subgraphs of unconstrained size $k$, Braverman et al. \cite{braverman2017eth} show that \ndks \ requires quasi-polynomial time to approximate by an additive constant assuming the Exponential Time Hypothesis.
Barman \cite{barman2015approximating} provides an additive quasi-PTAS for this setting via a sampling algorithm based on an approximate version of Caratheodory's theorem.
In the special case where $k = \Omega(n)$, Arora et al. \cite{arora1995polynomial} present an additive PTAS by estimating the coefficients to a polynomial program via sampling.

Another problem closely related to \ndks \ is the problem of finding approximate Nash Equilibrium in bimatrix games.
Barman extends his techniques for \ndks \ to achieve a quasi-PTAS for finding approximate Nash Equilibrium \cite{barman2015approximating}.
For small probability mixed Nash Equilibirum, in which the mixed strategy is supported by a linear number of pure strategies, Barman \cite{barman2015approximating} and Daskalakis et al. \cite{daskalakis2011oblivious} both provide a PTAS.
Both approximation schemes for small probability games fundamentally rely on sampling pure strategies.

\paragraph{Other Related Work.} Further lines of work concerned with incentivizing effort by others include algorithmic delegation \cite{kleinberg2018delegated,bechtel2020delegated,bechtel2022delegated} and algorithmic design of scoring rules \cite{hartline2020optimization,liu2023surrogate,chen2021optimal}.

\section{Preliminaries}
In this work, we study the principal-agent combinatorial contract design problem in both the single and the multi-agent setting introduced in \cite{duetting2022multi,dutting2022combinatorial}. 

\paragraph{The Single Agent Model}
In the single agent setting, a principal wants to delegate a task to an agent, resulting in an outcome from the set of all possible task outcomes $\Omega$. 
There is a ground set of actions $[n] = \{1,\dots , n\}$ and each action $i\in [n]$ has cost $c_i > 0$. 
The agent can choose to select any subset of actions $S \subseteq [n]$ and incurs an additive total cost of $c(S) = \sum_{i\in S} c_i$. 
The set of selected actions $S$ induces a probability distribution $q_S:\Omega \rightarrow [0,1]$ over the set of outcomes $\Omega$. The principal obtains reward $r(\omega)$ from outcome $\omega \in \Omega$.
In other words, if the agent selects a set of actions $S$ then the an outcome $\omega$ is drawn from the distribution $q_S$ and the principal receives reward $r(\omega)$. 

\paragraph{The Multi-Agent Model}
In the multi-agent model, the principal interacts with a set of agents $[n] = \{1,\dots ,n\}$.
We consider the case when each agent is assigned a single action.
Each agent has two strategies: to complete their assigned action or not.
If agent $i\in[n]$ acts, then they incur cost $c_i > 0$. 
Similar to the single agent setting, each subset of agents $S \subseteq [n]$ that exerts effort induces a probability distribution $q_S:\Omega \rightarrow [0,1]$ over the set of outcomes $\Omega$ and the principal obtains reward $r(\omega)$ from outcome $\omega \in \Omega$.

\paragraph{Reward Function}
In both settings, we define a reward function $f:2^{[n]} \rightarrow [0,1]$ as a mapping from the set of exerted actions to the expected reward of the principal.
We assume that the reward function is a monotone set function, i.e. $f(S) \geq f(T)$ for all $T \subseteq S$. We further assume that $f$ is normalized, i.e. $f(\emptyset) = 0$ and $f(S) \leq 1$ for all $S\subseteq [n]$.
For any $i\in [n]$ and $S \subseteq [n]$, we denote the marginal value of $i$ to the set $S$: $f(i\mid S) = f(S\cup i) - f(S)$. In this work, we focus on the following classes of reward functions:
\begin{enumerate}
    \item \textbf{Supermodular Functions:} We say that the set function $f:2^{[n]} \rightarrow \mathbb R_{\geq 0}$ is supermodular if it models increasing marginal returns.
    That is, for any sets $S \subseteq T$ and action $i\in [n]\setminus T$, $f(i\mid S) \leq f(i\mid T)$.
    \item \textbf{Graph Supermodular (Edge Density) Functions:} We say that the set function $f:2^V \rightarrow \mathbb R_{\geq 0}$ is graph supermodular if there exists an undirected graph $G=(V,E)$ such that for any $S\subseteq V$, $f(S) = \frac{|E(S)|}{\emax}$, where $|E(S)|$ is the number of edges in $G$ with both endpoints in $S$ and $\emax = \binom{n}{2}$ is the maximum number of edges in a graph with $n$ nodes.
    \end{enumerate}
    
\paragraph{Moral Hazard and Contracts}
The key challenge in contract theory is that the agent/agents (in both single and multi-agent settings) have no incentive to exert costly effort because only the principal receives reward for the completed task.
Therefore, the principal designs a \emph{contract} $t:\Omega \rightarrow \mathbb R^{n}_{\geq 0}$, which maps each outcome $\omega$ to a non-negative transfer $t_i(\omega)$ for the agent $i\in [n]$.

In this work, we focus on \emph{linear contracts}.
In the single agent setting, a linear contract is defined by a scalar $t \ge 0$, such that the agent receives payment $t(\omega) = t\cdot r(\omega)$ for any outcome $\omega\in\Omega$.
In the multi-agent setting, a linear contract is defined with a vector $(t_1,\dots, t_{n}) \in \mathbb{R}_{\ge 0}^{n}$ that determines the payment to each agent $i\in [n]$ for any outcome $\omega\in\Omega$: $t_i(\omega) = t_i \cdot r(\omega)$.
Linear contracts are an important class of contracts in practice because they are simple and provide desirable theoretical properties.
In the commonplace setting of binary outcomes (e.g., a task that can succeed or fail), linear contracts are optimal.
In general, they are known to be max-min optimal when only the expected reward of each set of agents is known \cite{dutting2019simple,duetting2022multi}. 

We assume that the outcome space is binary $\Omega = \{0, 1\}$ with rewards $r(0) = 0$ and $r(1) = 1$ W.L.O.G.
In our setting, the reward function $f$ can be interpreted as a probability of outcome $1\in\Omega$ or, equivalently, probability of successfully completing the underlying task.
Here, in the single agent setting, we can simplify the contract $t(\omega)$ for $\omega \in \Omega = \{0,1\}$ to scalar $t\in \mathbb R_{\geq 0}$ by assuming $t(0) = 0$ W.L.O.G. Similarly, for the multi-agent setting, we denote the contract by $(t_1,\dots , t_n) \in \mathbb R_{\geq 0}^n$.

Thus, in the single agent setting, if the agent performs the set of actions $S\subseteq [n]$ then the expected transfer to the agent is $t  \cdot f(S)$.
In the multi-agent setting, when the set of agents $S\subseteq V$ exert effort then the expected transfer to the agent $i\in [n]$ is $t_i \cdot f(S)$.  

\subsection*{Utility Functions, Oracles, and Equilibrium} 
\paragraph{The Single Agent Model}
In the single agent setting, given the set of actions $[n]$, reward function $f$, and linear contract $t$, the expected utility of the agent when they perform a set of actions $S\subseteq [n]$ is
\begin{equation} \label{eq:single_agent_utility}
    \mu_a(S, t) =  t\cdot f(S) - c(S), 
\end{equation}
The principal's expected utility under contract $t$ if the agent takes actions $S\subseteq [n]$ is
\begin{equation}\label{eq:single_principal_utility}
    \mu_p(S, t) = f(S) - t \cdot f(S) .  
\end{equation}
Thus, for a given linear contract $t$, the agent's best response is to select a set of actions $S\subseteq [n]$ that maximizes $\mu_a(S, t)$.
We assume that the agent breaks ties in favor of the principal: the agent will select a set with the greatest reward if there are multiple sets that provide her equal utility.
Hence, the optimal linear contract for the principal, assuming a rational agent, can be expressed as the value of $t$ maximizing the following optimization problem: 
\begin{flalign}
    & & \max_{t\ge 0, S\subseteq [n]} \quad &\mu_p(S, t) \nonumber \\
    & & \text{s.t.} \quad & \mu_a(S,t) \ge \mu_p(S',t),
    &\quad \forall S'\subseteq [n] \nonumber
\end{flalign}

In this setting, we assume access to a \textit{value oracle}, which we denote as $f$.
A value oracle computes reward $f(S)$ given an input $S\subseteq [n]$.
We assume access to a \textit{demand oracle}, which is commonly used in algorithmic game theory.
This oracle is a function $\phi$ that takes as input a vector of prices $p \in \mathbb{R}^n$ and outputs a set $S\subseteq [n]$ that maximizes $f(S) - \sum_{i \in S} p_i$.

The last oracle we use is an \emph{agent oracle} $\Phi$ that takes as input a contract $t$ and outputs the set of actions taken by a rational agent who observes tie-breaking rules (ties are broken in favor of action sets with higher reward).
Formally,
\begin{equation*}
    \Phi(t) \in \arg\max_{S \subseteq [n]} t\cdot f(S) - c(S)
\end{equation*}
The agent oracle $\Phi$ can be implemented efficiently with access to a demand oracle $\phi$ by using prices $p_i = c_i / t$ as input to the demand oracle ($c_i / t$ is infinity when $t = 0$).
Tie breaking behavior can be enforced by increasing the contract $t$ by an infinitesimally small amount.
Thus, under our assumption that we have access to a demand oracle, we also have access to an agent oracle.
We say $S$ is \textit{in demand} at contract $t$ if $S = \Phi(t)$ (tie-breaking is applied).
More generally, we say $S$ is \textit{in demand} if $\exists t$ such that $S = \Phi(t)$.

\paragraph{The Multi-Agent Model}
For a fixed set of agents $[n]$ each with a single action, a reward function $f$, and a linear contract $(t_1,\dots ,t_n )$, the expected utility of agent $i\in[n]$ is
\begin{equation}\label{eq:multi_agent_utility}
    \mu_a(S, t) = \begin{cases}
        f(S)\cdot t_i - c_i \quad &\text{ if } i\in S\\
        f(S)\cdot t_i \quad &\text{ otherwise }
    \end{cases}
\end{equation}

and the principal's expected utility is
\begin{equation} \label{eq:multi_principal_utility}
    \mu_p(S, t) = \left( 1 - \sum_{i\in[n]} t_i \right)\cdot f(S)
\end{equation}

We assume that agents act rationally and attempt to maximize their expected utility.
Notably, the action or inaction of one agent determines the payout of other agents.
Thus, as \cite{duetting2022multi, babaioff2006combinatorial} find, the agents partake in a game.
We specify that the agents play an ordinal potential game \cite{tijs2003introduction}.\footnote{The potential function for the agents' game is $f(S) - \sum_{i\in[n]} \frac{c_i}{t_i}$ for a fixed contract $t_1,\dots ,t_n$.}
A set $S$ of agents/actions forms a pure Nash Equilibrium if and only if:
\begin{align*}
  &f(S)\cdot t_i - c_i
  \geq f(S\setminus i)\cdot t_i
  && \forall i\in S\\
  &f(S)\cdot t_i
  \geq f(S \cup i)\cdot t_i - c_i
  && \forall i\notin S.
\end{align*}

Using the above equilibrium condition, \cite{duetting2022multi, babaioff2006combinatorial} find a closed form of the contract minimizing the total payment of the principal while incentivizing a set $S\subseteq [n]$ of agents to act. The best contract to incentivize $S\subseteq [n]$ if at all possible is
\begin{align*}
    &t_i = \frac{c_i}{f(i\mid S)} &\text{for all } i\in S \\
    &t_i = 0 &\text{for all } i\notin S
\end{align*}
Where $\frac{c_i}{f(i\mid S)}$ is infinity if $f(i\mid S) = 0$. With this payment scheme, the principal's objective is a set function:
\begin{equation}
    g(S) = \left( 1 - \sum_{i\in S} \frac{c_i}{f(i\mid S)} \right) f(S).
\end{equation}

The principal's problem is to maximize this objective.
Next, we define some notation which will be used throughout the paper. 

\paragraph{Graph Theory Notation} 
We denote an undirected graph as $G=(V,E)$ where $V$ is the set of nodes and $E$ is the set of edges.
The degree of a node $u\in V$ to a subgraph induced by $S \subseteq V$, denoted $\deg_S(u)$, is the number of neighbors $u$ has in $S$.
We denote the set of neighbors of a node $v\in V$ as $\mathcal{N}(v)$. 
The $k$-core of a graph $G$ is the unique maximal induced subgraph of vertices $S_k \subseteq V$ such that every node has degree at least $k$.

\section{Efficient Contracts for Single Agent Model} \label{Single_Agent_Section}
We explore the single agent combinatorial contracts problem presented by \dutting \cite{dutting2022combinatorial} and provide a strongly polynomial time algorithm to enumerate the breakpoints of the agent's utility function. 
Then, we focus on the setting in which the principal's reward function $f$ is supermodular and show that calculating the optimal contract in this setting is tractable.
This contrasts prior hardness results when $f$ is submodular \cite{dutting2022combinatorial}.  

\subsection{Problem Structure}
In the single agent model, the agent's utility at contract $t\in[0,1]$ is $\mu_a(t) = \max_{S \subseteq [n]} t\cdot f(S) - c(S)$.
It is easy to observe that the agent's utility w.r.t. contract $t$ is a convex and non-decreasing piece-wise linear function.
We define $T_S\subseteq [0,1]$ as the set of contracts for which $S$ is in demand.
It is known that the set $T_S$ is a left closed and right open interval, which follows from Observation~3.3 by \dutting \cite{dutting2022combinatorial}.
We let $\mathcal{D}_{f,c} = \{S\subseteq [n] : \exists t\in[0,1] \text{ such that } S=\Phi(t)\}$ contain all sets in demand.

If the principal was constrained to incentivize a specific set $S\in \mathcal{D}_{f, c}$, she would set the contract $\kappa_S = \min_{t \in T_S} t$ to minimize her payment to the agent.
We say that this contract, which is a \emph{break point} of the agent's piece-wise utility function, is the break point of set $S$.
Importantly, \dutting \cite{dutting2022combinatorial} show that the optimal contract is always a break point of the agent's utility function.

Our contribution relies on an important notion we call an intersection contract.
Formally, an intersection contract is the contract at which two sets of actions provide the agent equal utility.
For sets $L,R \subseteq[n]$, the intersection contract $\ic(L,R)$ is given by:
\[
    \ic(L, R) = \frac{c(L) - c(R)}{f(L) - f(R)}
\]
Intersection contracts are a generalization of break points.
Specifically, break points are the intersection contract of some pair of sets in the demand, but not all intersection contracts are break points.
To characterize some useful properties of intersection contracts, we define an ordering on sets in demand $\mathcal{D}_{f,c}$.

\begin{definition}[Ordering Sets in Demand]
    We define a strict total ordering $\prec$ on sets in demand as follows:
    Set $L$ precedes set $R$, denoted as $L\prec R$, means that $f(L) < f(R)$ or, equivalently, $\kappa_L < \kappa_R$.
\end{definition}

\begin{figure}[H]
\centering
\begin{subfigure}[t]{0.49\linewidth}
    \centering
    \includegraphics[width=\linewidth]{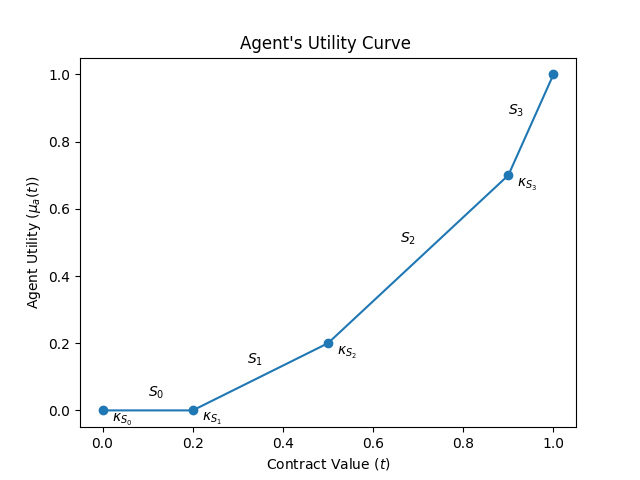}
    \caption{}
    \label{sub_fig:agent_curve}
\end{subfigure}\hfil 
\begin{subfigure}[t]{0.49\linewidth}
    \centering
    \includegraphics[width=\linewidth]{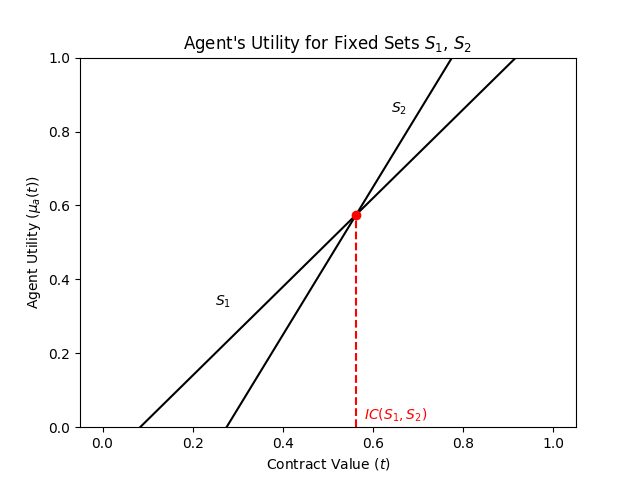}
    \caption[b]{}
    \label{sub_fig:intersection}
 \end{subfigure}
    \caption{
        Figure \ref{sub_fig:agent_curve} is an example of an agent's utility as a function of the contract $t$. 
        The sets $S_0, \dots, S_3$ are in demand with break points $\kappa_{S_0}, \dots, \kappa_{S_3}$, respectively. 
        Figure \ref{sub_fig:intersection} highlights the \textit{intersection contract} $\ic(S_1, S_2)$.
    }
\end{figure}

The following observations characterize the relationship between break points and intersection contracts.
The intersection contract of two adjacent sets in demand (sets $L,R\in \mathcal{D}_{f,c}$ are adjacent if there does not exist $S\in\mathcal{D}_f$ such that $L\prec S\prec R$) is a break point of the agent's piece-wise utility function.
Furthermore, the intersection contract of non-adjacent sets in demand yields a contract at which a different set is in demand.

\begin{observation} \label{adj_intersection}
    Let $L, R \in \mathcal{D}_{f,c}$ with $L \prec R$.
    The break point of $R$, denoted as $\kappa_{R}$, is the intersection contract of $L$ and $R$ if and only if $L$ and $R$ are adjacent. 
\end{observation}

\begin{observation}\label{lemma:break_point}
    Let $L$ and $R$ be two non-adjacent sets in demand such that $L\prec R$.
    Then, the set in demand $S$ at contract $t=\ic(L,R)$ satisfies $L\prec S \prec R$. 
\end{observation}

\subsection{Strongly Polynomial Time Algorithm for Poly-Size Break Point Sets}
We give a simple divide-and-conquer algorithm that reconstructs the agent's utility curve and finds all the breakpoints with $\Theta(|\mathcal{D}_{f,c}|)$ queries to a demand oracle and value oracle. 
This algorithm strengthens and generalizes prior work by \dutting \cite{dutting2022combinatorial}, in which a strongly polynomial algorithm was presented for gross substitutes reward functions and a weakly polynomial algorithm was presented for arbitrary reward functions.

For intuition, we let sets $L$ and $R$ be in demand.
To find all the breakpoints between $\kappa_L$ and $\kappa_R$ we first calculate $t = \ic(L,R)$.
If $t$ is the break point corresponding with $R$, then $L$ and $R$ are adjacent.
Otherwise, they are not adjacent and we find the set $S$ in demand at contract $t$ with $ L\prec S \prec R$ due to Observation~\ref{lemma:break_point}.
We divide our problem in two and recurse on set pairs $L, S$ and $S, R$.

\begin{algorithm} 
    \caption{\textsc{BreakPoint}($L, R$)}\label{CritVal}
    \textbf{Result:} A set of break points. \\
    \textbf{Input:} $L,R\subseteq [n]$ in demand with $f(L) < f(R)$. \\
    \textbf{Require:} Access to an agent oracle $\Phi$ (derived from a demand oracle), value oracle $f$, and cost function $c$.
    \begin{algorithmic}
        \State $t_S \gets \ic(L, R)$
        \State $S \gets \Phi(t_S)$
        
        \If{$f(S) = f(R)$}
            \State \Return $\{t_S\}$
        \Else
            \State \Return $\textsc{BreakPoint}(L,S) \cup \textsc{BreakPoint}(S,R)$
        \EndIf
    \end{algorithmic}
\end{algorithm}

\begin{theorem}[Single Agent Algorithm]\label{thm:single_agent_algo}
    Algorithm~\ref{CritVal} computes all break points with no more than $2|\mathcal{D}_{f,c}| + 1$ queries to a demand oracle and $\Theta(|\mathcal{D}_{f,c}|)$ queries to a value oracle.
    Thus, if $|\mathcal{D}_{f, c}| = \poly(n)$ and the demand and value oracles are implemented in strongly polytime, Algorithm~\ref{CritVal} finds the optimal contract in strongly polytime.
\end{theorem}
\begin{proof}
    We use induction on the \emph{distance} between $L$ and $R$, defined as $d_{L,R}:= |\{S \in \mathcal D_f : L \prec S \prec R \}|+1$.
    When $d_{L,R} = 1$, the sets $L$ and $R$ are adjacent. 
    Therefore, due to Observation~\ref{adj_intersection}, there is only one breakpoint $\ic(L,R) = \kappa_R$ which is found by Algorithm~\ref{CritVal}. 

    Now, suppose the algorithm finds all breakpoints between $\kappa_L$ and $\kappa_R$ (including $\kappa_R$) whenever $d_{L,R}\leq k$.
    Let $L,R$ be sets satisfying $d_{L,R}= k+1$.
    Due to Observation~\ref{lemma:break_point}, the demanded set $S = \Phi(\ic(L,R))$ satisfies $L \prec S \prec R$.
    Therefore $d_{L,S}, d_{S,R} \leq k$.
    Algorithm~\ref{CritVal} makes recursive calls and, by the inductive hypothesis, finds all break points in the interval $(\kappa_L, \kappa_R]$. 
    With inputs $L = \emptyset = \Phi(0)$ and $R=\Phi(1)$\footnote{
        The optimal contract is in the interval $[0,1]$, because contracts $t > 1$ yield the principal negative utility.
    }, Algorithm~\ref{CritVal} finds all breakpoints except the trivial break point $t = 0$ in $\Theta(|\mathcal{D}_f|)$ queries to the value and demand oracle.
\end{proof}

\subsection{Supermodular Contracts}
We show that when the reward function $f$ is supermodular, there are at most $n + 1$ sets in demand ($|\mathcal{D}_{f, c}|\le n+1$).
Combined with the fact that there is a strongly polytime demand oracle for supermodular functions \cite{iwata2001polysupermodular}, Theorem~\ref{thm:single_agent_algo} implies a strongly polytime algorithm for the single-agent supermodular contract problem.

Intuitively, when the set $S$ is in demand at contract $t$, every action $i\in S$ provides enough value to offset its cost $c_i$.
If the principal takes an additional action $j \in [n] \setminus S$ at a larger contract, the value of $i\in S$ offsets its cost $c_i$ even more because its marginal value increases due to supermodularity.
Thus, a rational agent will perform every action in $S$ at higher contracts.
Since there are only $n$ actions, there can only be at most $n + 1$ sets in demand.

In contrast, for submodular reward functions, the value of an action $i\in S$ may decrease with the addition of action $j\in [n] \setminus S$.
The value of action $i$ may no longer offset its cost $c_i$ even at higher contracts. This explains why there may be more than $n+1$ sets in demand for submodular reward functions \cite{dutting2022combinatorial}.

\begin{lemma} \label{thm:poly_sets_in_demand_single}
When $f$ is supermodular, there can be at most $n + 1$ sets in demand ($|\mathcal{D}_{f, c}| \leq n + 1$).
\end{lemma}
\begin{proof}
    Let $S_1$ and $S_2$ be sets in demand at contracts $t_1$ and $t_2$, respectively with $t_1\le t_2$.
    We show that $S_1 \subseteq S_2$.
    For contradiction, assume that $t_1 \le t_2$ and $S_1$ is not a subset of $S_2$. We show that $S_1 \cup S_2$  is in demand at contract $t_2$ to complete our contradiction:
    \begin{align}
        \mu_a(S_1 \cup S_2, t_2) &= f(S_1\cup S_2)t_2 - c(S_1\cup S_2) \nonumber \\
        &\ge [f(S_1) + f(S_2) - f(S_1\cap S_2)]t_2 - [c(S_1) + c(S_2) - c(S_1\cap S_2)] \label{sup_of_f} \\
        &= [f(S_2)t_2 - c(S_2)] + [f(S_1) - f(S_1\cap S_2)]t_2 - c(S_1) + c(S_1\cap S_2) \nonumber \\
        &= \mu_a(S_2, t_2) + [f(S_1) - f(S_1\cap S_2)]t_1 + [f(S_1) - f(S_1\cap S_2)](t_2 - t_1) - c(S_1) + c(S_1\cap S_2) \nonumber \\
        &= \mu_a(S_2, t_2) + \mu_a(S_1, t_1) - \mu_a(S_1\cap S_2, t_1) + [f(S_1) - f(S_1\cap S_2)](t_2 - t_1) \nonumber \\
        & > \mu_a(S_2, t_2) \label{high_t_high_f}
    \end{align}
    Line \ref{sup_of_f} holds by the supermodularity of $f$.
    Line \ref{high_t_high_f} leverages two main properties.
    First, $S_1$ is in demand at $t_1$, thus it provides more utility than $S_1\cap S_2$ at that contract: $\mu_a(S_1, t_1) \ge \mu_a(S_1\cap S_2, t_1)$.
    Second, $f(S_1)$ is greater than $f(S_1\cap S_2)$ by the monotonicity of $f$ and $t_2 \ge t_1$.
    This yields that $[f(S_1) - f(S_1\cap S_2)](t_2 - t_1) \ge 0$.

    The sets in demand must satisfy $S_1 \subset \dots \subset S_{|\mathcal{D}_{f, c}|}$. 
    Since there are $n$ actions from which the agent may choose and each demanded set must be a superset of some previous set, there are at most $n+1$ sets in demand.
\end{proof}

\newcommand{\Emax}{E_{max}}

\section{Graph Supermodular Multi-Agent Contract Problem}
In this section, we turn our focus to the Supermodular Multi-Agent Contracts Problem.
We begin by showing a strong impossibility result that rules out the existence of any polynomial time algorithm with bounded multiplicative approximations.
In addition, our impossibility result also rules out an additive FPTAS.
Our hardness result also holds even for uniform cost graph-supermodular contract problems. 

\subsection*{Hardness of Supermodular Multi-Agent Contract Problem}
Contrary to the submodular multi-agent contract setting \cite{duetting2022multi}, it is NP-hard to achieve a multiplicative factor approximation efficiently in the supermodular multi-agent contract problem even for instances in which agents have identical costs.
To prove this result, we reduce the decision version of $k$-clique (\kcd) (which is known to be NP-Hard) to the uniform cost graph supermodular multi-agent contract problem (\ugsc).

\begin{definition}[$k$-Clique Decision (\kcd)]
    $k$-Clique Decision is defined on an undirected graph $G=(V,E)$ and a positive integer $k$.
    The goal is to distinguish if the graph $G$ contains a $k$-clique, which is a subgraph of $k$ nodes with an edge between each pair of nodes.
\end{definition}
The U-GSC problem is a specific instance of the supermodular multi-agent contract problem in which the reward function is graph supermodular and costs are uniform for every agent.
Notably, when the reward function $f$ is the graph supermodular function defined on graph $G=(V,E)$, the marginal utility of agent $i\in V$ to the set $S\subseteq V$ is proportional to the degree of $i$ to $S$: $f(i\mid S) = \deg_{S}(i)/\emax$.

\begin{definition}[Uniform Cost Graph Supermodular Contract Problem (\ugsc)]
    The input to U-GSC is an undirected graph $G=(V,E)$ and a cost $c > 0$.
    The goal is to find a set $S\subseteq V$ maximizing the principal's utility for the reward function $f(S) := |E(S)|/\emax$ and cost $c$:
    \[
        \mu_p(S) = \left( 1 - \sum_{i\in S}\frac{c\cdot\emax}{\deg_{S}(i)} \right) \left( \frac{|E(S)|}{\emax} \right)
    \]
\end{definition}

Our reduction relies heavily on the properties of the sum of inverse degrees term (summation term).
An optimal solution $S^*\in\argmax_{S\subseteq V}\mu_p(S)$ to U-GSC will have low sum of inverse degrees, which implies $S^*$ will only contain nodes of high degree $\deg_{S^*}(i)$.
With a carefully selected cost $c$, we can ensure that only a $k$-clique (and its supersets) will contain nodes with the requisite large degree.

\begin{theorem}[Hardness of \ugsc]\label{thm:multi_agent_supermodular_hardness}
    The U-GSC problem admits no efficient multiplicative approximation algorithm nor an additive FPTAS unless $P=NP$.
\end{theorem}
\begin{proof}
    Given an instance of $k$-Clique Decision, $G=(V,E)$ and $k$, we use graph $G$ and cost $c=\frac{1}{\emax}\left(\frac{k-2}{k-1}\right)$ as input to the U-GSC problem.
    Trivially, this process takes polynomial time.
    
    In the following proofs we will show that the principal, in the constructed U-GSC instance, cannot receive utility more than $0$ if the \kcd \ instance is negative (i.e. does not contain a $k$-clique).
    For positive instances of \kcd, she can receive utility $\frac{k}{2\emax(k-1)} = \Omega\left(\frac{1}{n^2}\right)$.
    In \ugsc \ instances constructed by our reduction, a multiplicative approximation algorithm (FPTAS or PTAS) and an additive FPTAS for U-GSC will be able to find sets with positive utility.
    Thus, such approximation algorithms can distinguish between positive and negative instances of $k$-clique and do not exist unless $P=NP$.

    To start, we assume that the given instance of \kcd \ is positive and show that the principal can receive utility $\Omega \left(\frac{1}{n^2} \right)$ in the constructed instance of \ugsc.
    The graph $G$ contains a $k$-clique we denote as $K$.
    If the principal incentivizes the nodes $K$ in the constructed instance of \ugsc, then she receives utility
    \begin{align*}
        \mu_p(K)
        &= \left(1 - \left(\frac{k-2}{k-1}\right)\sum_{i\in K} \frac{1}{\deg_K(i)}\right) \left(\frac{|E(K)|}{\Emax}\right) \\
        &= \left(1 - \left(\frac{k-2}{k-1}\right)\left(\frac{k}{k-1}\right)\right) \left(\frac{k(k-1)}{2\Emax}\right) \\
        &= \frac{k}{2\Emax (k-1)}
    \end{align*}

    Next, we assume that the \kcd \ instance is negative and show that the principal cannot hope to get positive utility in the constructed instance of \ugsc.
    To show this, we leverage Turan's Graph Theorem \cite{aigner1995turan} to bound the number of edges in any subgraph absent of a $k$-clique.
    This theorem states that any subset of nodes $S$ from a graph without a $k$-clique satisfies the following:
    \begin{equation*}\label{eq:turans}
        |E(S)| \le \left(\frac{k-2}{k-1}\right) \left(\frac{|S|^2}{2}\right)
    \end{equation*}
    
    For a fixed set $S$, the sum of inverse degrees is minimized, and thus the principal's utility is maximized, when $\deg_S(i) = 2|E(S)|/|S|$\label{eq:handshake} by the handshaking lemma and convexity.
    With a bound on the number of edges in $S$ and the degree of each node in $S$, we bound the utility the principal can receive for any set of nodes $S$ absent of a $k$-clique.
    We omit the right term $|E(S)|/\Emax$ because it is inconsequential.
    \begin{align*}
        \mu_p(S)
        &\le \left(1 - \left(\frac{k-2}{k-1}\right)\sum_{i\in S} \frac{1}{\deg_S(i)}\right) \\
        &\le\left(1 - \left(\frac{k-2}{k-1}\right)\sum_{i\in S} \left(\frac{|S|}{2|E(S)|}\right)\right) 
        & (\text{handshaking lemma and convexity})\\
        &\le\left(1 - \left(\frac{k-2}{k-1}\right)\sum_{i\in S} \frac{(k-1)}{(k-2)|S|}\right)
        & (\text{Turan's Theorem}) \\
        &\le 0
    \end{align*}
    Thus, the principal cannot hope to receive positive utility for instances of \ugsc \ constructed from a negative instance of \kcd.

    We reiterate that any constant multiplicative approximation scheme for \ugsc \ can distinguish between positive and negative instances of \kcd.
    Additionally, an additive FPTAS finds a solution within $O\left(\frac{1}{n^3}\right)$ of the optimal in polynomial time, so it also distinguishes between positive and negative instances of \kcd \ efficiently.
    Thus, none of these approximation algorithms exist unless $P=NP$.        
\end{proof}

\subsection*{Multi-Agent Contracts and Dense Subgraphs}

We formalize a connection between \ugsc \ and Normalized Densest $k$-subgraph (\ndks) as formulated in Barman~\cite{barman2015approximating}: given an undirected graph $G=(V,E)$ and integer $k \in \mathbb N_{> 0}$, find $S\subseteq V$ of size $k$ that maximizes the objective $\frac{ |E(S)|}{|S|^2}$ \footnote{
    There are two notions of edge density that are commonly used: $|E(S)|/|S|^2$ and $|E(S)|/\binom{|S|}{2}$.
}.
Specifically, \ugsc \ subsumes a special case of \ndks \ we call ($\delta$, $\epsilon$)-linear sized almost clique (\lsac) in which we are promised that a graph $G$ either contains a $\delta n$-clique or no subgraph with size $\delta n$ has edge density ($\frac{|E(S)|}{|S|^2}$) more than $1 - \epsilon$.
The goal is to distinguish between these two cases.

\begin{proposition}(Informal)
    An additive PTAS for \ugsc \ implies a polytime algorithm for \lsac.
\end{proposition}

We provide a log-space reduction that implies an additive PTAS for \ugsc \ yields a poly-time algorithm for \lsac. 
Given a graph $G$, we set costs to be $c = \frac{(1-\epsilon)}{\emax}$.
The principal can receive utility $\Theta(\epsilon \delta^2)$ when $G$ contains a $\delta n$-clique.
In the case that every subgraph of size $\delta n$ has density no more than $(1-\epsilon)$, the principal cannot hope to receive utility more than $\Theta((1-\epsilon)\epsilon \delta^2)$.
Thus an additive PTAS for \ugsc \ with error set to $\Theta(\epsilon^2 \delta^2)$ can distinguish between negative and positive instances of \lsac \ in poly-time.\footnote{We note that Arora et al.~\cite{arora1995polynomial} present an additive PTAS for \ndks \ for linear sized $k$ that solves \lsac.}
For completeness, we leave additional details of this reduction in Appendix~\ref{sec:connection to DkS}.

\section{Additive PTAS for the U-GSC Problem} \label{sec:PTAS}

We turn our focus to obtaining an additive PTAS for the uniform cost graph supermodular contracts problem (\ugsc).
Recall the principal's objective in \ugsc:
\[
    \mu_p(S)
    = \left( 1 - \sum_{i \in S} \frac{c}{\deg_{S}(i)} \right) \left(\frac{|E(S)|}{\emax}\right)
\]
For convenience, we omit the $\emax$ term in each agent's marginal utility by implicitly factoring out a $\frac{1}{\emax}$ term from the cost (the reparameterized cost is $c\in[0, 1)$).
We denote the left term of the objective as $L(S) := 1 - \sum_{i \in S}\frac{c}{\deg_{S}(i)}$,\footnote{We occasionally refer to the summation $\sum_{i \in S} \frac{c}{\deg_{S}(i)}$ as the \textit{sum of inverse degrees}.} the right term of the objective as $R(S) := \frac{|E(S)|}{\emax}$, and let $S^*\in\argmax_{S\subseteq V}\mu_p(v)$ denote an optimal solution.
We notice that both $L(S), R(S)\leq 1$ for all $S\subseteq V$.
The goal of our additive PTAS for U-GSC is to find a set $S$ that obtains $\mu_p(S) \geq \mu_p(S^*) - O(\epsilon)$ in time $n^{O(1/\epsilon)}$ for $\epsilon \in (0, \frac{1}{7}]$.
We assume that the graph $G$ contains $n \ge \frac{2^2\cdot 3^4}{\epsilon^8}$ nodes.
Otherwise, we can find an optimal contract in time $n^{O(1/\epsilon)}$ by enumerating all subsets of nodes.
When $\mu_p(S^*) < \epsilon$, the empty set is within $\epsilon$ additive error of the optimal solution: $\mu(\emptyset) > \mu_p(S^*) - \epsilon$.
Hence, throughout this section, we focus on the case when $\mu_p(S^*) \ge \epsilon$.
In this nontrivial case, it immediately follows that $S^*$ contains a linear portion of the nodes in $V$.
\begin{restatable}{proposition}{linsizeopt}\label{prop:linear-size-opt-set}
    If optimal set $S^*$ satisfies $\mu_p(S^*) \ge \epsilon$, then $|S^*| \ge \epsilon \cdot n$. \label{linearS} 
\end{restatable}
We leave the proof of Proposition~\ref{prop:linear-size-opt-set} in the Appendix~\ref{sec:PTAS}.
Proposition~\ref{prop:linear-size-opt-set} implies that when the size of the optimal contract satisfies $|S^*| \leq \epsilon \cdot n$, we can simply output the empty set which will be within an $\epsilon$ additive factor of $\mu_p(S^*)$.
Hence, from now on, we assume $|S^*| \geq \epsilon \cdot n$.
\subsection*{\ndks \ Fails for \ugsc}
In the last section, we observed that \ugsc \ is intimately related to the $\ndks$ problem when $k$ is linear. An additive PTAS for linear \ndks \ has been shown by Arora et al.~\cite{arora1995polynomial}.
By Proposition \ref{prop:linear-size-opt-set}, we know that $S^*$ is of linear size.
This begs the question: why doesn't the ``obvious'' PTAS that simply tries the densest $k$ subgraph for every value of $k$ work?
Indeed, we know that $S^*$ must be dense to increase $R$, and dense sets tend to have many nodes of high degree (increasing $L$), so it is reasonable to think that $S^*$ may be well approximated by a densest $k$ subgraph for some $k = \Omega(n)$.
We show \ndks \ does not immediately imply an additive PTAS for \ugsc \ and we present an example (details in Appendix~\ref{sec:DkSfails}) in which for any $k$, there is a constant gap in utility between the optimal solution to \ndks \ and the optimal solution to \ugsc.

\begin{restatable}{example}{exampleoffail}\label{ex:dks_fails}
    Consider an undirected graph with $n$ nodes (where $n$ is a multiple of $12$).
    The graph is disjoint with two connected components $D, S\subseteq V$ that are each composed of half the nodes.
    In the subgraph $D = H \cup P$, $\frac{1}{6}n$ nodes (called $H$) form a $\frac{1}{6}n$-clique (have an edge to all the other $\frac{1}{6}n$ nodes).
    The other $\frac{1}{3}n$ nodes (called $P$) in $D$ have an edge to every node in the $\frac{1}{6}n$-clique.
    The subgraph $S$ is a complete bipartite subgraph with $2$ parts of $\frac{1}{4}n$ nodes each.
    When costs are set to $c = \frac{1}{4}$ in this example, the principal's optimal utility is at least $\mu_p(S^*) \geq \mu_p(S) \geq \frac{1}{16}$ and every densest $k$ subgraph $K$ yields utility less than $\frac{3}{50}$.
\end{restatable}

\vspace{.2 cm}
Example \ref{ex:dks_fails} exploits the fact that \ndks \ only tries to optimize the number of edges selected, with no regard to the degree distribution of the selected nodes.
In our constructed instance, the subgraph $K$ with the most edges always contains a ``small'' clique (for every $k$), thereby never achieving ``roughly even'' degree distribution for large enough $k$. 
Thus, for larger values of $k$, $L(K)$ is always smaller than $L(S^*)$ by an absolute constant, yielding a constant gap in utility from the optimal solution.\footnote{When $|K|$ is much smaller than $|S^*|$ (small values of $k$), $L(K) > L(S^*)$, but $K$ is simply too small to contain a comparable number of edges.}
Thus, a solver for \ndks \ does not immediately imply an additive PTAS for \ugsc.
Moreover, we note that the techniques for an additive PTAS in \cite{arora1995polynomial,barman2015approximating,feige1997densest} do not appear to generalize to our problem.
These techniques rely on polynomial program relaxations, which cannot express the sum of inverse degrees term in the \ugsc \ objective.

\subsection{Building Blocks of Our Approach} \label{sec:building block for PTAS}
An advantage of a linear sized optimal set is that when a node $v\in V$ is sampled uniformly at random it has constant probability of being in the optimal set ($\Pr[v\in S^*] \ge \epsilon$).    
Hence, if $O(\log n)$ nodes are sampled uniformly at random, then all the nodes lie in the set $S^*$ with probability at least $\epsilon ^{O(\log n)}$.
If the aforementioned sampling procedure is repeated $\poly(n)$ times, then for at least one of the runs, all the $O(\log n)$ nodes lie in the set $S^*$ with probability at least $1 - \frac{1}{\poly (n)}$.
This process, which we refer to as the $O(\log n)$-sampler in Lemma~\ref{lem:sampler}, efficiently samples $O(\log n)$ nodes uniformly from the set $S^*$.
A similar sampler was deployed to construct a PTAS for finding an approximate equilibrium of small probability games by Daskalakis and Papadimitriou \cite{daskalakis2011oblivious}.
The proof of Lemma~\ref{lem:sampler} is similar to that of Lemma 1 by Daskalakis et al. \cite{daskalakis2011oblivious}, but for completeness we prove the lemma in the Appendix.

\begin{restatable}{lemma}{efficientsampler}\label{lem:sampler}
    Let $K\subseteq V$ be a subset of nodes with $|K|\geq \frac{1}{3}\epsilon \cdot n$.
    Let $M_1,\dots, M_t$ for $t =  n^{\Theta\left(\frac{\log \frac{1}{\epsilon}}{\epsilon^4} \right)}$ be multisets of size $\Theta(\log n)$ sampled uniformly and at random from $V$ with replacement.
    With probability at least $1 - \frac{1}{\poly(n)}$, there exists a set $M_i$ for some $i\in [t]$ such that $M_i$ is a uniformly random multiset of $K$.
\end{restatable}

The key advantage of the $O(\log n)$-sampler is that for a high degree node $v\in V$ with $\deg_{S^*}(v) = \Omega (n)$, the $O(\log n)$-sampler provides an accurate estimate of $\deg_{S^*}(v)$ in polynomial time.
In the next lemma, we show that for any set $K$ of size $\Omega(n)$, we can estimate $\deg_K(v)$ for all nodes $v\in V$ with $\deg_K(v) \geq \Omega(n)$. The proof of the lemma is in the appendix.  

\begin{restatable}{lemma}{degreeestimates}\label{lem:key-sampling}
    Let  $\beta = \frac{1}{3}(1-\epsilon)\epsilon$ and $K\subseteq V$ be any subset of nodes with $|K|\geq \frac{1}{3}\epsilon \cdot n$.
    There is a $n^{\Theta\left( \frac{\log \frac{1}{\epsilon}}{\epsilon^4}\right)}$-time algorithm unaware of the nodes $K$ that computes $t = n^{\Theta\left( \frac{\log \frac{1}{\epsilon}}{\epsilon^4}\right)}$ different degree estimates $\est^1_K(v),\dots, \est^t_K(v)$ for each node $v\in V$ such that with probability at least $1-\frac{1}{\poly(n)}$ that satisfies   there exists an $i\in[t]$.
       
    \begin{enumerate}
        \item  $\Pr\left[ \widehat{\deg}^i_K(v) < \beta n \right]\geq 1 - \frac{1}{2 n^3}$ for all $v\in V$ with $\deg_{K}(v) < \frac{\epsilon \cdot \beta }{9}n$
        
        \item $\Pr\left[  (1-\epsilon) \cdot \deg_K(v) \leq \widehat{\deg}^i_K(v) \leq  (1+\epsilon) \cdot \deg_K(v) \right]\geq 1 - \frac{1}{2  n^3}$ for all $v\in V$ with $\deg_{K}(v) \geq \frac{\epsilon\cdot \beta}{9}n$
    \end{enumerate}
\end{restatable}

Suppose for a moment that all the nodes in $v\in S^*$ have $\deg_{S^*}(v) = \Omega(n)$.
Lemma~\ref{lem:key-sampling} implies that with probability at least $1 - \frac{1}{\poly(n)}$, for all $v\in S^*$, $\est_K(v)$ is within a $(1 \pm \epsilon)$-multiplicative factor of $\deg_K(v)$.
Hence, with probability at least $1 - \frac{1}{\poly(n)}$ the sum of degree estimates do not effect the sum of inverse degrees too much, 
\begin{align*}
    \sum_{v\in S^*} \frac{c}{\est_{S^*} (v)}
    \leq \frac{1}{(1-\epsilon)}\cdot \sum_{v\in S^*} \frac{c}{\deg_{S^*} (v)}
    \leq \sum_{v\in S^*} \frac{c}{\deg_{S^*} (v)} + 2 \epsilon,
\end{align*}
where the last inequality holds because $\sum_{v\in S^*} \frac{c}{\deg_{S^*} (v)} \leq 1$ and $\epsilon <\frac 1 2$.
This implies that $L(S^*)$ is approximated within an $O(\epsilon)$ additive factor.
However, Lemma~\ref{lem:key-sampling} does not ensure the estimates $\est_{S^*}(v)$ within $(1  \pm \epsilon)$-factor of $\deg_{S^*}(v)$ when $\deg_{S^*}(v)\leq \frac{\epsilon \cdot \beta }{9}n$. 
Thus, we cannot bound the sum of inverse degrees when we sample from $S^*$.

To overcome this issue, we pose a question.
Is there an approximately optimal set $\tilde S \subseteq S^*$ such that every node in $\tilde S$ has linear degree, $\tilde S$ contains a linear number of nodes, and $\tilde S$ is approximately optimal (i.e. $\deg_{\tilde S}(v) = \Omega(n)$ for all $ v\in V$, $|\tilde S| \ge \frac{1}{3}\epsilon n$, and $\mu_p(\tilde S) \ge \mu_p(S^*) - \epsilon$)?
If such a set $\tilde S$ exists, then by \ref{lem:key-sampling} we can compute degree estimates that do not substantially affect the sum of inverse degrees: $\sum_{v\in \tilde S} \frac{c}{\est_{\tilde S} (v)}\leq  \sum_{v\in \tilde S} \frac{c}{\deg_{\tilde S} (v)} + 2 \epsilon.$

A natural approach to construct such a set $\tilde S$ is to begin with $\tilde S = S^*$ and iteratively remove nodes with small $\deg_{\tilde S}(v)$ until only nodes of high degree remain (similarly to the classic coring algorithm).     
However, as the low-degree nodes are removed, the degree of the rest of the nodes (with potentially $\Omega(n)$ degree in $S^*$) can decrease significantly, possibly increasing the sum of inverse degrees substantially.
We overcome this obstacle by analyzing a relaxation of our objective for every iteration of the classic coring algorithm.
In the process, we show the existence of a set $\tilde S\subseteq S$ with the desired properties.

\subsection*{Existence of Approximately Optimal $\tilde S \subseteq S^*$ with Large Degrees}

The key idea is to show that the $O(\epsilon)\cdot n$-core of $S^*$ is non-empty and approximately optimal.
The $k$-core of a graph is the (unique) maximal subgraph such that every node has degree at least $k$.
One finds the $k$-core of a graph by iteratively removing the minimum degree node as long as there is a node with degree less than $k$.

Next, we show that the $\epsilon n/3$-core of the optimal set $S^*$ contains a linear number of nodes and provides near optimal utility.
\begin{lemma}\label{lem:solid-core}
    Let $S^*$ be the optimal contract with $|S^*|\geq \epsilon n$.
    There exists a subset $\tilde S \subseteq S^*$ such that $\deg_{\tilde S} (v) \geq \frac{\epsilon n}{3}$ for all $v\in \tilde S$, $|\tilde S| \ge \frac{1}{3}\epsilon n$, and $\mu_p(\tilde S)\geq \mu_p(S^*) - O(\epsilon)$. 
\end{lemma}

For convenience, we denote the $\epsilon n/3$-core of $S^*$ as $\tilde{S}$.
Trivially, $\tilde S$ will only contain nodes with degree at least $\frac{\epsilon n}{3}$.
We use a counting argument to bound the number of nodes and edges in $\tilde S$.
Lastly, we bound $L(S) - L(\tilde{S})$ by analyzing the change in a relaxed version of $L$ each iteration of the coring algorithm.
Before finishing the proof of this lemma, we introduce several other lemmas formalizing the above ideas.

\begin{lemma}\label{lem:core_right}
    Let $S^*$ be an optimal contract with $|S^*| \ge \epsilon n$.
    Let $\tilde S$ be the $\epsilon n/3$-core of the optimal set $S^*$.
    Then, $|E(S^*)| - |E(\tilde S)| \le \frac{1}{3}\epsilon n^2$.
    This implies that $|\tilde{S}| \ge \epsilon n/3$ and $R(S) - R(\tilde S) \le \epsilon$ for graphs with $n \geq 3$ nodes.
\end{lemma}
\begin{proof}
    Recall that the discrete coring algorithm iteratively removes nodes that do not meet the $\frac{1}{3}\epsilon n$ degree requirement, implying that at most $\frac{1}{3}\epsilon n$ edges are removed each iteration.
    Since at most $n$ nodes are removed, the coring algorithm removes at most $\frac{1}{3}\epsilon n^2$ edges: $|E(S^*)| - |E(\tilde S)| \le \frac{1}{3}\epsilon n^2$.

    To prove that $\tilde S$ is non-empty, we use the previous fact and the property $\frac{1}{3} \epsilon n^2 < \frac{\epsilon n (n-1)}{2}$ for $n > 3$.
    In the nontrivial case when $\mu_p(S^*) \geq \epsilon$, it must be that $|E(S^*)| \ge \frac{\epsilon n (n-1)}{2}$.
    Thus, $|E(\tilde S)| \ge |E(S^*)| - \frac{1}{3} \epsilon n^2 \ge \frac{\epsilon n (n-1)}{2} - \frac{1}{3} \epsilon n^2 > 0$.
    This implies that $\tilde S$ is non-empty and, since $\tilde S$ is a $\frac{1}{3} \epsilon n$-core, it must contain at least $\frac{1}{3} \epsilon n$ nodes.

    Lastly, it holds that $R(S^*) - R(\tilde S) \le \frac{\frac{1}{3}\epsilon n^2}{\emax} \le \epsilon$ for graphs of $n \geq 3$ nodes.
\end{proof}

\begin{lemma}\label{lem:core_left}
    Let $S^*$ be an optimal contract with $|S^*| \ge \epsilon n$.
    Let $\tilde S$ be the $\epsilon n/3$- core of the optimal set $S^*$.
    It holds that $L(S^*) - L(\tilde{S}) \leq \epsilon$.
\end{lemma}
\begin{proof}
    We define a relaxation for the left half of the objective: $\mathcal{L}(S) := 1 - \sum_{i\in S} \frac{c}{\deg_{S}(i) + 1}$.
    Clearly, for all sets $S\subseteq V$, $\mathcal{L}(S) \ge L(S)$.
    Consider the change in $\mathcal{L}$ for one iteration of the coring algorithm in which a node $v \in S$ is removed from the set of remaining nodes $S$.
    This node satisfies $\deg_{S}(v) \le \deg_{S}(i)$ for all $i\in S$.
    \begin{align*}
        \mathcal{L}(S\setminus \{v\}) - \mathcal{L}(S)
        &= \left(1 - \sum_{i\in S\setminus\{v\}}\frac{c}{\deg_{S\setminus\{v\}}(i) + 1}\right) - \left(1 - \sum_{i\in S}\frac{c}{\deg_{S}(i) + 1}\right) \\
        &= \frac{c}{\deg_{S}(v) + 1} - \sum_{i\in S\cap\mathcal{N}(v)} \frac{c}{\deg_{S}(i) (\deg_{S}(i) + 1)} \\
        &\geq \frac{c}{\deg_{S}(v) + 1}\left(1 - \sum_{i\in S\cap\mathcal{N}(v)}\frac{1}{\deg_{S}(v)}\right)
        &\left(\deg_{S}(v) \le \deg_{S}(i), \forall i\in S\right) \\
        &= 0
        &\left(\sum_{i\in S\cap \mathcal{N}(v)} 1 = \deg_{S}(v)\right)
    \end{align*}

    Next, we bound the gap between the left objective and its relaxed variant for the $\epsilon n/3$-core.
    \begin{align*}
        L(\tilde S) - \mathcal{L}(\tilde S)
        &= \left(1 - \sum_{i\in S}\frac{c}{\deg_{S}(i)}\right) - \left(1 - \sum_{i\in S}\frac{c}{\deg_{S}(i) + 1}\right) \\
        &= c\sum_{i\in S}\left(\frac{1}{\deg_{S}(i) + 1} - \frac{1}{\deg_{S}(i)}\right) \\
        &\ge cn \left( \frac{1}{\frac{1}{3}\epsilon n + 1} - \frac{1}{\frac{1}{3}\epsilon n} \right)
        &(\deg_{S}(i) \ge \frac{1}{3}\epsilon n, \forall i\in S) \\
        &= -\frac{cn}{(\frac{1}{3}\epsilon n + 1)(\frac{1}{3}\epsilon n)} 
        \ge -\frac{9c}{\epsilon^2 n}
        \ge -\frac{\epsilon^6}{36}
        &(c \le 1\text{ and }n\ge \frac{2^2\cdot 3^4}{\epsilon^8})
    \end{align*}

    The above inequalities imply that
    \begin{align*}
        L(S^*) \le \mathcal{L}(S) \le \mathcal{L}(\tilde{S}) \le L(\tilde{S}) + \epsilon
    \end{align*}
\end{proof}

\begin{proof}[Completed Proof of Lemma~\ref{lem:solid-core}]
    We know that $L(S^*) - L(\tilde{S}) \le \epsilon$ and $R(S^*) - R(\tilde{S}) \le \epsilon$ by lemmas \ref{lem:core_left} and \ref{lem:core_right}, respectively.
    With the above, we bound the difference in utility for $S^*$ and $\tilde{S}$    
    \begin{align*}
        \mu_p(S^*) - \mu_p(\tilde S)
        = L(S^*)\cdot R(S^*) - L(\tilde S)\cdot R(\tilde S)  
        \leq  L(S^*)\cdot R(S^*) - (L(S^*) - \epsilon)\left( R(S^*) - \epsilon \right)  \leq  2\epsilon
    \end{align*}
    Lastly, Lemma~\ref{lem:core_right} shows that $\tilde{S}$ contains the required $\epsilon n/3$ nodes.
\end{proof}

\subsection{The Main Algorithm and Additive PTAS}

In Section~\ref{sec:building block for PTAS}, we established the following properties of the optimal contract:
\begin{enumerate}
    \item \label{property1} In Lemma~\ref{lem:solid-core}, we proved an existence of $\tilde S \subseteq S^*$ with the properties $|\tilde S| \geq \frac{\epsilon}{3} \cdot n$, $\forall v\in \tilde S$, $\deg_{\tilde S}(v)\geq \frac{\epsilon}{3} \cdot n$, and $\mu_p(\tilde S) \geq \mu_p(S^*) - O(\epsilon)$.
    \item \label{property2}By Lemma~\ref{lem:key-sampling}, we can compute degree estimates $\est_{\tilde{S}}(v)$ satisfying $(1-\epsilon)\deg_{\tilde S}(v) \le \est_{\tilde{S}}(v) \le (1+\epsilon)\deg_{\tilde S}(v)$ for all $v\in \tilde{S}$ with probability at least $1 - \frac{1}{\poly(n)}$ without knowledge of $\tilde{S}$.\footnote{
        Lemma~\ref{lem:key-sampling} produces many degree estimates for each node.
        With high probability at least one estimate will be accurate for every node.
        We assume W.L.O.G. that we can identify an accurate degree estimate if one exists, as we can simply enumerate over all estimates.
    } 
\end{enumerate}
Property~\ref{property1} shows the existence of $\tilde S$ which is almost optimal such that every node $v\in\tilde S$ has high degree $\deg_{\tilde S}(v) = \Omega(n)$.
Instead of searching for the contract that approximates $\mu_p(S^*)$, we aim to find a contract that approximates $\mu_p(\tilde S)$.
Since $\tilde S$ has a linear number of nodes and only contains nodes of linear degree ($|\tilde S| = \Omega(n)$ and $\deg_{\tilde S}(v) = \Omega(n), \forall v\in \tilde S$), property~\ref{property2} allows us to ``accurately" estimate the degrees $\deg_{\tilde S}(v)$ for all nodes $v\in \tilde S$, even without any description of $\tilde S$. 
    
The first step of our PTAS is to estimate $\deg_{\tilde S}(v)$ for all $v\in V$.
We then filter out the nodes with with low estimated degree.
By property~\ref{property2}, this filtering process retains all nodes in $\tilde{S}$ in addition to other nodes with accurate degree estimates.    

To formalize the above discussion, given the degree estimates of all nodes, we consider the set of high degree nodes denoted to $\tilde S$ as $H$ such that
\begin{equation*}
    H = \left\{ v\in V: \est_{\tilde S}(v) \geq \frac{(1-\epsilon)\epsilon \cdot n}{3} \right\}
\end{equation*}
We let $\beta = \frac{(1-\epsilon)\cdot \epsilon}{3}$.
Now, Lemma~\ref{lem:key-sampling} implies that with probability at least $1 - \frac{1}{\poly  (n)}$, $H$ does not contain any node $v\in V$ with $\deg_{\tilde S}(v)\leq \frac{\epsilon \cdot \beta}{9} n$.
Hence, for the rest of this section, we condition on the following event whose probability is at least $1 - \frac{1}{\poly(n)}$:
\begin{equation*}
    \mathtt{GoodSample}
    = \{ \tilde S \subseteq H\}
    \cap \left \{ \forall v \in H: \deg_{\tilde S}(v) > \frac{\epsilon\cdot\beta }{9} n \right\}
    \cap \left\{ \forall v\in H: (1-\epsilon)\leq \frac{\est_{\tilde S}(v)}{\deg_{\tilde S}(v)} \leq (1 + \epsilon)  \right\}
\end{equation*}

\begin{observation}\label{obsn:high-prob-event}
    $\Pr[\mathtt{GoodSample}] \geq 1 - \frac{1}{\poly(n)}$.
\end{observation}
\begin{proof}
    By definition, $H$ only contains nodes of high approximate degree to $\tilde S$ ($\est_{\tilde{S}}(v) \ge \beta n$).
    Nodes $v$ with low true degree $\deg_{\tilde{S}}(v) < \frac{\epsilon \cdot \beta}{9} n$ do not have high approximate degree by Sampling Lemma~\ref{lem:key-sampling}, so they are not included in $H$.
    Also by Sampling Lemma~\ref{lem:key-sampling}, all other nodes in $H$ must have accurate degree (within a $(1\pm\epsilon)$ factor).
    Finally, all nodes in $\tilde S$ must be $H$ as their true degree is high.
    Lemma~\ref{lem:key-sampling} states that all of the mentioned events occur with probability $1 - \frac{1}{\poly(n)}$.
\end{proof}

The rest of our algorithm focuses on the nodes in $H$ to reconstruct $\tilde{S}$ because $\tilde S \subseteq H$ when $\mathtt{GoodSample}$ occurs.
We note that the remainder of our proofs condition on the event $\mathtt{GoodSample}$.

To recover $\tilde S$ from the set $H$ we set up a tractable optimization problem.
Preferably, we want to find a set $S \subseteq H$ that minimizes $\sum_{v\in S} \frac{c}{\deg_S(v)}$ subject to $|E(S)| \ge |E(\tilde{S})|$.
The minimizer $S$ of this optimization problem will have high utility when $\mathtt{GoodSample}$ occurs: $\mu_p(S) \geq \mu_p(\tilde{S}) \geq \mu_p(S^*) - O(\epsilon)$.
Unfortunately, this optimization problem appears to be challenging to solve optimally, as it is not conducive to the quadratic programming or linear programming techniques of \cite{arora1995polynomial, barman2015approximating, feige1997densest}.
    
Therefore, we define a surrogate linear program that utilizes accurate estimates of $\deg_{\tilde S}(v)$ obtained in Lemma~\ref{lem:key-sampling}.  
We define weights $\{x\}_{v\in H}$, where $x_v$ is the indicator of $v \in \tilde S$.
For ease of notation, we denote $\hat d_v = \est_{\tilde S}(v)$. 
\begin{flalign}
    &  &\mathtt{LP}(H,E(\tilde S), \{\hat d_v\}_{v\in H}):  \min_{\{x_v\}_{v\in H}} \quad &\sum_{v\in H} \frac{c}{\hat d_v}  \cdot x_v \label{lp:1_objective}\\
    & & \text{s.t.} \quad & \sum_{v\in H}\hat{d}_v x_v \ge 2(1-\epsilon) \cdot |E(\tilde S)| \label{lp:1_sufficient_edges}\\
    & &      & \sum_{\substack{u\in H\\ (u,v)\in E}} x_u \geq \left(\frac{1}{1+\epsilon}\right) \cdot \hat d_v, &\forall v\in H \quad \label{lp:1_degrees_lowerbound}\\
    & &      & x_v\in[0,1]   &\forall v\in H \quad \nonumber
\end{flalign}

In this LP, we minimize the fractional sum of inverse estimated degrees.
The first constraint \eqref{lp:1_sufficient_edges} ensures that our fractional solution contains enough edges.
Finally, constraint \eqref{lp:1_degrees_lowerbound} ensures that the fractional degree of a node is greater its estimated degree with the allowance of some error.

An issue with the surrogate LP is that the algorithm is unaware of the set $\tilde S$, so we do not know the value of $|E(\tilde S)|$.
However, there are only $O(n^2)$ many values of $|E(\tilde S)|$ which can be efficiently enumerated, so we assume that $|E(\tilde S)|$ is known W.L.O.G.
Next, we obtain the following lemma that shows that the fractional solution to the above LP approximates the contract $\tilde S$ whenever $|E(\tilde S)|$ is given and $\mathtt{GoodSample}$ occurs.

\begin{observation}
    Let $\tilde{\vec{x}} = \mathbbm{1}_{\tilde S}$ be the indicator of $\tilde S$.
    Observe that $\tilde{\vec{x}}$ is a feasible solution to LP~\eqref{lp:1_objective}.
    Thus, the optimal solution $\vec{x}^*$ to LP~\eqref{lp:1_objective} has lower sum of inverse approximate degrees: $\sum_{v\in H}\frac{c}{\hat{d}_v}\tilde{\vec{x}} \le \sum_{v\in H}\frac{c}{\hat{d}_v}\vec{x}^*$.
\end{observation}

The last step of our PTAS is to employ a randomized rounding scheme to recover a set $S$ that is ``near optimal'' using the fractional solution $\vec{x}^*$.
This rounding scheme includes each node $v$ in the constructed set $S$ independently with probability $\vec{x}^*_v$.
We denote $S\sim \mathtt{Round}(\vec{x}^*)$ to mean that $S\subseteq H$ is the result of rounding fractional solution $\vec{x}^*$.
Intuitively, if $S\sim\mathtt{Round}(\vec{x}^*)$ then for any $v\in H$, we have $\mathbb E[\deg_{T}(v)] = \sum_{\substack{u\in H\\ (u,v)\in E}} x_u = \Omega(n)$ due to Constraint~\ref{lp:1_degrees_lowerbound},  and $\hat d_v \geq \beta \cdot n$ for all $v\in H$.
Therefore, with high probability $\deg_{T}(v)$ should concentrate around $\hat d_v$.

We essentially utilize this fact and obtain the following lemma that rounds the optimal solution of LP~\eqref{lp:1_objective} to an integral contract which is almost optimal.
Our rounding techniques are similar to derandomization techniques from \cite{arora1995polynomial}.
We present the proof in Appendix~\ref{alg:round}.

\begin{restatable}{lemma}{roundinglemma}\label{lem:round}
    Let $\vec x^*$ be the optimal solution to LP~\eqref{lp:1_objective} and $S \sim \mathtt{Round}(\vec{x}^*)$.
    Then with probability at least $1 - \frac{1}{\poly(n)}$, we have $\mu_p(S) \geq \mu_p(S^*) - 5 \epsilon$.
\end{restatable}

\subsection*{Putting it All Together: }
Next, we outline our PTAS for U-GSC in Algorithm~\ref{alg:PTAS_subprocess}.

\begin{theorem}\label{thm:PTAS}
    For any given constant $\epsilon \in \left(0 ,\frac{1}{7} \right]$ and instance of \ugsc \ $\langle G=(V,E) , c \rangle$, Algorithm~\ref{alg:PTAS_subprocess} outputs a contract $S$ such that with probability at least $1 - \frac{1}{\poly(n)}$ contract $S$ has high utility: $\mu_p(S) \geq \mu_p(S^*) - O(\epsilon)$.
\end{theorem}
\begin{proof}
    By Observation \ref{obsn:high-prob-event}, we obtain a good sample that accurately estimates the degree of every node $v \in H$ to $\tilde S$ and has $\tilde S \subseteq H$ with high probability.
    Then, given a guess for $|E(\tilde S)|$ (by enumeration), the resulting LP will be feasible.
    Let $\vec{x}^*$ denote the optimal solution to the feasible LP.
    By Lemma \ref{lem:round}, $S$ sampled according to $\vec{x}^*$ ($S \sim \mathtt{Round}(\vec{x}^*)$) will satisfy $\mu_p(S) \geq \mu_p(S^*) - 5\epsilon$.
\end{proof}

\begin{algorithm}
\caption{Additive PTAS for Uniform Cost Graph Supermodular Contracts} \label{alg:PTAS_subprocess}
\textbf{Input:} An error parameter $\epsilon\in\left(0,\frac{1}{7}\right]$ and an undirected graph $G=(V,E)$ with $n\ge \frac{2^2\cdot 3^4}{\epsilon^8}$ nodes. \\
\textbf{Output:} A set $S$ satisfying $\mu_p(S^*) - \mu_p(\tilde S) \le 7\epsilon$ With probability $1 - \frac{1}{n}$.
\begin{algorithmic}
    \State $\texttt{Contracts} \gets \{\emptyset\}$
    \For{ $i=1,\dots, n^{\left(\frac{3^7}{\epsilon^4} \log\left(\frac{3}{\epsilon}\right)\right)}$ }      
        \State Let $M(i)$ be a multiset of $m=\frac{3^6}{\epsilon^4} \ln(n)$ nodes drawn uniformly at random from $V$ (with replacement).
        \State For each node $v \in V$, define the degree estimate $\hat{d}_v(i) \gets \min\left(\frac{|K|}{m}\deg_{M(i)}(v), n\right)$.
        \State Define the set of high degree nodes $H(i) \gets \left \{ v\in V \mid \hat{d}_v(i) \ge \frac{1}{3}(1-\epsilon)\epsilon n \right\}$.
        \For{Each possible value of $|E(\tilde S)|$: $k = \frac{\epsilon^2\cdot n^2}{18}, \dots, n^2$}
            \If{$\mathtt {LP}(H(i), k , \{\hat d_v(i)\}_{v\in H(i)})$ defined in (\ref{lp:1_objective}) is feasible}
                \State Let $\vec x^*(i,k)$ be the optimal solution to the $\mathtt {LP}(H(i), k , \{\hat d_v(i)\}_{v\in H(i)})$. 
                \State $S(i,k) \sim \mathtt{Round}(\vec x^*(i,k))$,
                \State $\texttt{Contracts} \gets \texttt{Contracts} \cup  S (i,k)$
            \EndIf
        \EndFor
    \EndFor
    \State $\hat S\gets \argmax_{S\in \mathtt{contracts}} \mu_p(S)$.
    \State \Return $\tilde S$
\end{algorithmic}
\end{algorithm}


\section{Open Problems and Future Directions}
We reiterate an important open question posed by \dutting \cite{dutting2022combinatorial} in the single agent setting: is there a polynomial time algorithm to find the optimal contract for arbitrary reward functions given access to a demand and value oracle?
There are two known classes of reward functions for which we can efficiently compute an optimal contract, both of which admit an efficient demand oracle.
These are gross substitutes as shown in \dutting \cite{dutting2022combinatorial}, and supermodular functions as shown in this work.
Both \cite{dutting2022combinatorial} and this work rely on there being a polynomial number of sets in demand.\footnote{Interestingly, the known families of set functions with a polynomial number of sets in demand all admit efficient demand oracles.}
However, there may be an exponential number of sets in demand in general so a new approach is required.

We leave open the existence of an additive PTAS for graph supermodular functions with arbitrary costs.
Recall that in the uniform cost case, low degree nodes are not necessary in an approximate solution and our PTAS discards such nodes entirely.
The arbitrary cost case poses additional difficulty because low degree nodes are still desirable to the principal if they have low cost.
Using our PTAS in the arbitrary cost setting may yield a poor solution if there are many low degree nodes with low cost.
Other approaches we considered and believe to be promising leverage regularity of large graphs.
For instance, Szemer\'edi's Regularity Lemma ensures a partition on nodes exhibiting random-like properties that appear useful for degree analysis in the multi-agent contract problem.

Our last open problem is to design an additive PTAS in the mutli-agent setting with arbitrary supermodular functions (that are normalized and monotone) with arbitrary costs.
Our analysis of the setting with graph supermodular rewards and uniform costs suggests this problem is related to constrained supermodular maximization.

\section{Acknowledgements}
We thank Aviad Rubinstein for helpful discussions and pointers to related $\dks$ literature. We also thank Shang-Hua Teng for helpful discussions.

\bibliographystyle{plain}
\bibliography{references}
\appendix
\section{DkS Fails for U-GSC} \label{sec:DkSfails}
Intuitively, a high utility set in the contract problem cannot contain a single node of low degree because the sum of inverse degrees will be too large (summation term).
Thus, the principal requires that every selected node has large degree.
In contrast, the $\dks$ problem is only concerned with density (the number of edges among selected nodes), and not individual degree.
To concretely show the difference of these conflicting goals, consider the following \ugsc \ instance for which the densest $k$-subgraph for any $k = 1,\dots, n$ does not provide a PTAS.

\exampleoffail*
\begin{proof}
    For ease of notation, write $\delta \in [0, 1]$ as the portion of nodes that must be selected (i.e. $k = \delta n$). 
    The output $K$ of the \dks \ solver can be partitioned into four cases. 
    When $\delta \leq 1/6$, $K$ will be a subset of the $n/6$-clique ($K \subseteq H$) achieving the maximum possible density of $1$. Thus, for $\delta \leq 1/6$, $\mu_p(K)$ is:
    \begin{align*}
        \left( 1 - \delta n \left( \frac{1/4}{\delta n} \right) \right) \left( \frac{2}{n^2} \left( \frac{(\delta n)^2}{2} \right) \right) = \frac{3 \delta ^2}{4} < 0.03  \qquad \text{for } \delta \leq 1/6  
    \end{align*}
    When $1/6 < \delta \leq 1/2$, it can be shown that the densest $k$ subgraph is all of $H$ and some nodes in $P$. Thus, for $1/6 < \delta \leq 1/2$, $\mu_p(K)$ is:
    \begin{align*}
        \left(1 - \frac{n}{6} \left( \frac{1/4}{\delta n} \right) - \left( \delta - \frac{1}{6} \right) n \left( \frac{1/4}{n/6} \right) \right) \left( \frac{2}{n^2} \left( \frac{(n/6)^2}{2} + \left( \delta - \frac{1}{6} \right) n \left( \frac{n}{6} \right) \right) \right) \\
        = \frac{-432 \delta^3 + 396 \delta^2 - 42 \delta + 1}{864 \delta} < 0.06 \qquad \text{for } 1/2 < \delta < 2/3
    \end{align*}
    When $1/2 < \delta < 2/3$, $K$ runs out of nodes to select from $H$ and $P$, so it begins adding nodes from $S$. Thus, for $1/2 < \delta < 2/3$, $\mu_p(K)$ is:
    \begin{align*}
        \left(1 - \frac{n}{6} \left( \frac{1/4}{n/2} \right) - \left( \frac{n}{3} \right) \left( \frac{1/4}{n/6} \right) - \left( \delta - \frac{1}{2} \right) n \left( \frac{1/4}{\frac{(\delta - 1/2)n}{2}} \right) \right) \left( \frac{2}{n^2} \left( \frac{(n/6)^2}{2} + \frac{n}{3} \left( \frac{n}{6} \right) + \left( \frac{ (\delta - 1/2) n}{2} \right)^2  \right) \right) \\
        = -\frac{36 \delta^2 - 36 \delta + 19}{864} < 0 \qquad \text{for } 1/2 < \delta < 2/3
    \end{align*}
    At $\delta = 2/3$, $K$ drops all nodes in $P$ and replaces them with the entirety of $S$ ($K = H \cup S$). So when $\delta \geq 2/3$, $K$ will select all of $H$ and $S$ and fill in the remaining nodes with nodes in $P$. For $2/3 < \delta \leq 1$, $\mu_p(K)$ is:
    \begin{align*}
        \left(1 - \frac{n}{6} \left( \frac{1/4}{(\delta - 1/2)n} \right) - \left( \delta - \frac{2}{3} \right) n \left( \frac{1/4}{n/6} \right) - \left( \frac{n}{2} \right) \left( \frac{1/4}{n/4} \right) \right) \left( \frac{2}{n^2} \left( \frac{(n/6)^2}{2} + \left( \delta - \frac{2}{3} \right) n \left( \frac{n}{6} \right) + \left( \frac{n/2}{2} \right)^2  \right) \right) \\
        = \frac{\left( -36 \delta^2 + 54 \delta - 19 \right) \left( 24 \delta - 5 \right)}{864 \left(2 \delta - 1 \right)} < 0.04
        \qquad \text{for } 2/3 \leq \delta \leq 1
    \end{align*}
    To reiterate, the principal's utility for incentivizing $K$ (for any $k$) is at most $\mu_p(K) < 0.06 < \frac{1}{16}$.
\end{proof}

\section{Multi-Agent U-GSC and Dense Graphs}\label{sec:connection to DkS}
We first define an \lsac \ instance which is a restricted version of verification problem from \cite{braverman2017eth}.
\begin{definition}[$(\delta,\epsilon)$-linear sized almost-clique (\lsac)]
    We are given an undirected graph $G=(V,E)$ of $n$ vertices and constants $\delta, \epsilon > 0$.
    The goal of \lsac \ is to distinguish between the following two cases when $k=\delta n$:
    \begin{enumerate}
        \item $G$ contains a $k$-clique.
        \item Every subgraph of $G$ with at least $k$ nodes has density at most $1-\epsilon$ $\left(\text{i.e., }\frac{|E(S)|}{|S|^2} \le 1-\epsilon\right)$.
    \end{enumerate}
\end{definition}
Notably, without the assumption that $k$ is linear in $n$, the above problem requires at least $n^{\tilde{\Omega}(\log n)}$ time to solve assuming the Exponential Time Hypothesis ($ETH$) \cite{braverman2017eth}.
We will show that our PTAS for U-GSC implies that \lsac \ admits a polytime algorithm. 

\begin{proposition}(Informal)
    An additive PTAS for \ugsc \ problem implies a polytime algorithm for \lsac \ (assuming $\epsilon$ is a constant).
\end{proposition}

\begin{proof}
    We provide a reduction from \lsac \ to \ugsc \ in which the principal obtains $\Theta(\poly(\delta, \epsilon))$ more utility if and only if the \lsac \ instance contains a $k$-clique.
    Thus, a PTAS for \ugsc \ can be used to distinguish between the two cases of the \lsac \ problem in polynomial time.

    Given an instance of \lsac, we construct an instance of \ugsc \ with input graph $G$ and cost $c = (1-\epsilon)/\emax$.
    In positive cases of \lsac \ (i.e. $G$ contains a $\delta n$-clique) the principal will receive utility at least $\Omega(\epsilon \delta^2)$.
    We denote the $\delta n$-clique as $K$ and calculate the principal's utility if she incentivizes the agents in $K$ to act:
    \begin{align*}
        \mu_p(K) 
        &= \left(1 - (1 - \epsilon) \sum_{i\in S} \frac{1}{\deg_S(i)}\right) f(K) \\
        &= \left(1 - (1 - \epsilon) \frac{|S|}{|S| - 1}\right) f(K) \\
        &\ge \epsilon f(K) \\
        &= \Omega (\epsilon\delta^2)
    \end{align*}        

    For a negative instance of \lsac \ without a $\delta n$-clique, we consider two cases.
    In the first case, the principal pays at least $\delta n$ agents and in second case, she pays no more than $\delta n$ agents.
    We show that in both cases, the principal will receive utility no more than $O(\epsilon (1-\epsilon) \delta^2)$.
    This means that the principal receives $\Theta(\epsilon^2\delta^2)$ more utility for positive instances of \lsac.
    
    Let $S\subseteq V$ be the set of agents that the principal pays in the first case ($|S| \ge \delta n$).
    Since the graph does not contain a subgraph of size $\delta n$ with density greater than $1-\epsilon$, it must be that $S$ cannot have density greater than $1-\epsilon$ or some random sample of $S$ would have greater than $1-\epsilon$ density. Thus,
    \begin{align*}
        \frac{|E(S)|}{|S|^2} &\le (1-\epsilon) \\
        |E(S)| &\le (1-\epsilon) |S|^2
    \end{align*}
    As stated previously, the sum of inverse degree term is minimized when the degree is evenly distributed over every agent ($\deg_S(i) = \frac{2|E(S)|}{|S|}, \forall i\in S$).
    Using these two properties, we show that the principal gets no more than $0$ utility when she incentivizes the set of agents $S$ with $|S|\ge\delta n$:
    \begin{align*}
        \mu_p(S) 
        &= \left(1 - (1 - \epsilon) \sum_{i\in S} \frac{1}{\deg_S(i)}\right) f(S) \\
        &\le \left(1 - (1 - \epsilon) \sum_{i\in S} \frac{|S|}{2|E(S)|}\right) f(S)  &(\text{even degree}) \\
        &\le \left(1 - (1 - \epsilon) \frac{2|S|^2}{(1-\epsilon) |S|^2}\right) f(S) 
        &\text{(bound on edge count)} \\
        &\le 0
    \end{align*}

    Now consider the second case, in which the principal incentivizes agent set $S$ to take action with $|S|<\delta n$.
    It holds that $|E(S)|\le (1-\epsilon) |E(K)|$ where $|E(K)|$ is the number of edges a $\delta n$-clique would have, meaning $f(S) \le (1-\epsilon) f(K)$.
    We use the property that $2|E(S)|\le |S|^2$ and the previously mentioned property that the principal's utility is maximized when degrees are evenly distributed over agents ($\deg_S(i) = \frac{2|E(S)|}{|S|}, \forall i \in S$).
    The principal's utility if she incentivizes less than $\delta n$ agents is upper bounded as follows:
    \begin{align*}
        \mu_p(S)
        &\le \left(1 -  \frac{(1 - \epsilon) \cdot |S|^2}{2|E(S)|}\right) f(S) \le \left(1 -  \frac{(1 - \epsilon)\cdot |E(S)|}{|E(S)|}\right) f(S) = \epsilon f(S) \le \epsilon (1-\epsilon) f(K)\leq \epsilon \cdot \delta^2 ( 1- \epsilon).
    \end{align*}
    With this reduction, the principal can get utility $\epsilon f(K)$ if the \lsac \ instance contains a clique, otherwise she cannot derive more than $\epsilon (1-\epsilon) f(K)$ utility.
    This implies that a PTAS for \ugsc \ additive error set to $\Theta(\epsilon^2 \delta^2)$ can solve the linear sized almost clique problem efficiently.
\end{proof}
\section{Missing Proofs from Section~\ref{sec:PTAS}}
\linsizeopt*
\begin{proof}
    For the sake of contradiction, let $|S^*| \le \epsilon n$. Then $E(S^*) \le \frac{\epsilon\cdot n \cdot (\epsilon \cdot n - 1)} 2 $. This implies that,

    \begin{equation*}
        \mu_p(S^*) \leq R(S^*) = \frac{|E(S^*)|}{\emax} \le  \frac{\epsilon\cdot n(\epsilon\cdot n-1)}{n(n-1)} \le \epsilon \qedhere
    \end{equation*} 
\end{proof}

\subsection{Existence of Sampler}
\efficientsampler*
\begin{proof}
    We find the probability that a multiset $M$ of $\frac{3^6}{\epsilon^4} \log n$ nodes sampled uniformly at random from $V$ independently with replacement only contains nodes in the set $K$:
    \begin{align*}
        \Pr[v \in K \text{ for all } v\in M]
        = \left( \frac{|K|}{n} \right)^{\left( \frac{3^6}{\epsilon^4} \log n \right)}
        \geq \left(\frac{\epsilon}{3}\right)^{\left( \frac{3^6}{\epsilon^4} \log n \right)}
        = n^{\left( \frac{3^6}{\epsilon^4} \log \left(\frac{\epsilon}{3}\right) \right)}
    \end{align*}
    If $t=n^{\left(\frac{3^7}{\epsilon^4}\log \left(\frac{3}{\epsilon}\right)\right)}$ different multisets $M_1,\dots, M_t$ of size $\frac{3^6}{\epsilon^4} \log n$ are drawn uniformly and independently from $V$, then at least one only contains nodes from $K$ with probability $1 - \frac{1}{n^3}$:
    \begin{align*}
        \Pr[\text{exists } M_i \text{ with nodes only from $K$}]
        &= 1 - \Pr[\text{does not exist } M_i \text{ with nodes only from $K$}] \\
        &= 1 - (1 - \Pr[v \in K \text{ for all } v\in M_i])^{t} \\
        &= 1 - \left(1 - n^{\left( \frac{3^6}{\epsilon^4} \log \left(\frac{\epsilon}{3}\right) \right)}\right)^{n^{\left(\frac{3^7}{\epsilon^4}\log \left(\frac{3}{\epsilon}\right)\right)}} \\
        &\ge 1 - \frac{1}{n^3}
    \end{align*}
    
    Finally, we show that this sampler process yields a uniform sample over nodes in $K$.
    If $v$ is sampled $V$ uniformly at random, then $v$ is drawn uniformly at random from $K$ when conditioning for $v\in K$. Consider an arbitrary node $u\in K$:
    \begin{align*}
        \Pr[u = v\mid v\in K]
        = \frac{\Pr[(u = v)\cap (v\in K)]}{\Pr[v\in K]}
        = \frac{\Pr[u = v]}{\Pr[v\in K]} 
        = \frac{\left(\frac{1}{|V|}\right)}{\left(\frac{|K|}{|V|}\right)}
        = \frac{1}{|K|}
    \end{align*}
    Thus, we can efficiently find a sample of $O(\log (n))$ nodes drawn uniformly at random from $K$ with high probability even without knowledge of the set $K$. 
\end{proof}

\degreeestimates*
\begin{proof}
    As $|K| =\Omega(n)$, Lemma~\ref{lem:sampler} ensures the existence of efficient sampler which samples $m \ge\frac{3^6 \cdot \ln(n)}{\epsilon^4}$ many nodes uniformly at random from the set $K$ with replacement.
    The rest of the proof has two parts.
    To start, we consider a low degree node $v\in V$ with $\deg_{K}(v) < \frac{1}{9}\epsilon\beta n\le \frac{\beta}{3}|K|$ (use the fact that $|K|\ge\epsilon n/3$).
    The random variable $X$ is the number of neighbors node $v$ has in the multi set $M$ of $m$ nodes.
    The expected value of $X$ is $\mu = \frac{\deg_{K}(v)}{|K|}m$.

    \begin{align*}
        \Pr\left[ \frac{|K|}{m}X \ge\beta n \right]
        &\le \Pr\left[ \frac{|K|}{m}X \ge \beta |K| \right] && (|K|\leq n)  \\
        &= \Pr\left[ X \ge \beta m \right]\\
        &= \Pr[X\ge(1+\eta)\mu] && \left(\text{by setting } \eta = \frac{\beta m}{\mu} - 1\right) \\
        &\leq \exp\left(-\frac{\eta^2 \mu}{\mu + 2}\right)
    \end{align*}
    Set the above to less than $\frac{1}{2n^3}$, substitute in for $\eta$ and $\mu$, and isolate $m$ to yield:

    \begin{align*}
        m
        \ge \frac{\left(\frac{\beta|K|}{\deg_{K}(v)} + 1\right)}{ \beta \left(\frac{\beta|K|}{\deg_{K}(v)} - 2\right)} \ln(2n^3)
        \ge \frac{4\cdot \ln(2n^3)}{\beta}
    \end{align*}

    The last inequality holds because $\deg_{K}(v)\le\frac{\epsilon^3\beta}{3}|K|$ and the function $\frac{x+1}{x-2}$ is decreasing in $x$ when $x\geq 3$ and $\beta \leq \epsilon/3$.
    The second part of this proof shows that sampling yields good degree estimates on nodes with high degrees.
    A node $v\in V$ has high degree if $\deg_{K}(v) \ge \frac{1}{9}\epsilon\beta n$.
    Again, we let $X$ be a random variable representing the number of neighbors node $v$ has in the multi-set $M$ and let $\mu = \frac{\deg_{K}(v)}{|K|}m$ be the expected value of $X$.
    \begin{align*}
        \Pr\left[\left| \frac{|K|}{m}X - \deg_{K}(v) \right| > \epsilon \deg_{K}(v) \right]
        &=\Pr[|X - \mu| \ge \epsilon\mu] \le 2\exp\left(-\frac{\epsilon^2\mu}{3}\right) \\
        &= 2 \cdot \exp\left(-\frac{\epsilon^2 \deg_{K}(v) m}{3|K|}\right)\\
        &\leq 2 \cdot \exp\left(-\frac{\epsilon^3 \cdot \beta \cdot  m}{3^3|K|}\right)\leq \frac 1 {2n^3}.
    \end{align*}
    The second to last inequality holds because $|K|\le n$ and $\deg_{K}(v)>\frac{\epsilon\beta}{9} \cdot  n$ and the last inequality holds because  $m \ge\frac{3^6 \cdot \ln(n)}{\epsilon^4}$.
    A sample of size $m\ge\frac{3^6\ln(n)}{\epsilon^4}$ is sufficient for accurate degree approximates of low and high degree nodes.
    Applying union bounds over all $n$ nodes, the sample yields an accurate approximated well with a probability of at least $1/2n^2$.
\end{proof}

\subsection{Proof of Rounding Scheme}\label{alg:round}
\roundinglemma*
\begin{proof}
        We define the following desired events:
    \begin{enumerate}
        \item $\mathtt{OBJ}$ is the event when $\sum_{v\in S} \frac{c}{\hat{d}_v} < \sum_{v\in H} \frac{c \cdot \vec{x}^*_v}{\hat{d}_v} + \epsilon$.
        \item  $\mathtt{EDGE}$ is the event when $\sum_{v\in S}\hat{d}_v > (1-\epsilon)\sum_{v\in H}\hat{d}_v \vec{x}^*_v$
        \item For $v\in H$, $\mathtt{Deg}(v)$ is the event when $\deg_{S}(v) \geq (1-\epsilon) \sum_{\substack{u\in H\\(u,v)\in E}}\vec{x}^*_v$
    \end{enumerate}
    First,  we observe that if for the set $S$, event $\mathtt{OBJ_i}\cap \mathtt{EDGE_i} \bigcap_{v\in H} \mathtt{Deg}(v) \cap \mathtt{GoodSample}$ occurs then,
    \begin{align*}
        |E(S)|
        = \frac{1}{2}\sum_{v\in S}\deg_S(v)
        \ge \frac{1}{2}\left(\frac{1-\epsilon}{1+\epsilon}\right)\sum_{v\in S} \hat{d}_v
        \ge \frac{1}{2}\frac{(1-\epsilon)^2}{(1+\epsilon)}\sum_{v\in H} \hat{d}_v \vec{x}^*_v
        \ge \frac{(1-\epsilon)^3}{(1+\epsilon)} |E(\tilde S)|, 
    \end{align*}
    where the first inequality holds due to the event $\bigcap_{v\in H}\mathtt{Deg_i}(v)$ and $\mathtt{GoodSample}$, the second inequality holds due to the event $\mathtt{EDGE}$ and the last inequality holds due to constraint~\ref{lp:1_sufficient_edges}.
    Moreover, 
    \begin{align*}
            \sum_{v\in S}\frac{c}{\deg_S(v)}
            &\le \left(\frac{1+\epsilon}{1-\epsilon}\right)\sum_{v\in S}\frac{c}{\hat{d}_v}
            \le \left(\frac{1+\epsilon}{1-\epsilon}\right)\epsilon + \left(\frac{1+\epsilon}{1-\epsilon}\right)\sum_{v\in H}\frac{c \cdot \vec{x}^*_v}{\hat{d}_v} \\
            &\le \left(\frac{1+\epsilon}{1-\epsilon}\right)\epsilon + \frac{1+\epsilon}{(1-\epsilon)^2} \cdot \sum_{v\in \tilde S}\frac{c}{\deg_{\tilde S}(v)},
    \end{align*}
    where, the first inequality holds due to the event $\mathtt{OBJ}$ and the second inequality holds due to the fact that $\vec x^*$ is the optimal solution to LP and $\tilde{\vec x} = \mathbbm 1_{\tilde S}$ is feasible when conditioned on the event $\mathtt{GoodSample}$.
    \begin{align*}
        \mu_p(S)
        &= \left( 1-\sum_{v\in S}\frac{c}{\deg_S(v)} \right)\left(\frac{|E(S)|}{E_{max}}\right) \\
        &\ge \left( 1 - \left( \frac{1 + \epsilon}{1 - \epsilon} \right) \epsilon - \frac{1 + \epsilon}{(1 - \epsilon)^2} \sum_{v \in \tilde S} \frac{c}{\deg_{\tilde S}(v)} \right) \left( \frac{(1-\epsilon)^3}{(1 + \epsilon)} \frac{ |E(\tilde S)|}{ \emax} \right) \\
        &\ge \left( 1-\sum_{v\in \tilde S }\frac{c}{\deg_{\tilde S}(v)} \right)\left(\frac{|E(\tilde S)|}{E_{max}}\right) - 5\epsilon \\
        &\ge \mu_p(\tilde S) - 5\epsilon.
    \end{align*}
    Therefore, for the rest of the proof, we focus on proving that with probability at least $1 - \frac{1}{\poly(n)}$, the event $\mathtt{OBJ_i}\cap \mathtt{EDGE_i} \bigcap_{v\in H} \mathtt{Deg}(v) \cap \mathtt{GoodSample}$ occurs.
    
    First, we condition on the event $\mathtt{GoodSample}$. Next, we apply Hoeffding's bound to bound the probability that $\mathtt{OBJ}$ complement occurs:
    \begin{align*}
        \Pr[\mathtt{OBJ}^c \mid \mathtt{GoodSample}]
        &= \Pr\left[\sum_{v\in S}\frac{c}{\hat{d}_v}
        \ge \sum_{v\in H}\frac{c\cdot\vec{x}^*_v}{\hat{d}_v} + \epsilon\right]
        = \Pr\left[\sum_{v\in S}\frac{c}{\hat{d}_v}
        \ge \mathbb{E}\left[\sum_{v\in S}\frac{c}{\hat{d}_v}\right] + \epsilon\right] \\
        &\le \exp\left(-\frac{2\epsilon^2}{\sum_{v\in H} \left(c/\hat{d}_v\right)^2}\right)\\
        &\leq \exp\left(-\frac{2\epsilon^2 }{\sum_{v\in H} \left(1/\hat{d}_v\right)^2}\right)
        \leq  \exp\left(-\frac{2\epsilon^2 }{ \sum_{v\in H} \left(\frac 1 {\beta\cdot n}\right)^2}\right)\\
        &=\exp\left(-\frac{2\epsilon^2 \cdot \beta^2 \cdot n^2}{|H|}\right)\leq \exp\left(-{2\epsilon^2 \cdot \beta^2 \cdot n }\right) .
    \end{align*}
    Above, the first inequality holds due to Hoeffding.
    The second inequality holds because $c \in (0,1)$.
    The third inequality holds because under event $\mathtt{GoodSample}$, $ \hat d_v \geq \beta \cdot n$.
    The last inequality holds because $|H|\leq n$.
    
    The probability that $\mathtt{EDGE}$ complement occurs can be bounded using Chernoff bounds:
    \begin{align*}
        \Pr[\mathtt{EDGE}^c \mid \mathtt{GoodSample}]
        &= \Pr\left[\sum_{v\in S}\hat{d}_v
        \le (1-\epsilon)\sum_{v\in H}\hat{d}_v \vec{x}^*_v\right]
        = \Pr\left[\sum_{v\in S}\hat{d}_v
        \le (1-\epsilon)\mathbb{E}\left[\sum_{v\in S}\hat{d}_v \right]\right] \\
        &\le\exp\left( -\frac{2\epsilon^2 \mathbb{E}\left[ \sum_{v\in S}\hat{d}_v \right]^2}{n\left(\max_{v\in H} \hat{d}_v \right)^2} \right)
        \le\exp\left( -\frac{2\epsilon^4 n}{9} \right).
    \end{align*}
    Above, we use the fact that $\mathbb{E}\left[\sum_{v\in S} \hat{d}_v \right] = \sum_{v\in H}\hat{d}_v \vec{x}^*_v \ge 2(1-\epsilon)|E(\tilde S)|\ge|E(\tilde S)|\ge\frac{1}{3}\epsilon n^2$ due to constraint~(\ref{lp:1_sufficient_edges}) and Lemma~\ref{lem:solid-core}.  

    Next, the probability that $\mathtt{Deg}(v)$ complement occurs can be bounded using Chernoff bounds.
    We have,
    \begin{align*}
        \Pr\left[\deg_S(v)
        \le (1-\epsilon)\sum_{\substack{u\in H\\(u,v)\in E}}\vec{x}^*_v  \mid \mathtt{GoodSample}\right]
        &= \Pr\left[\deg_S(v)
        \le (1-\epsilon)\mathbb{E}\left[\deg_S(v)\right] \right]  \\
        &\le\exp\left( -\frac{\epsilon^2 \mathbb{E}\left[\deg_S(v)\right]}{2} \right)
        \le\exp\left( -\frac{\epsilon^2 \beta n}{4} \right)
        \le\exp\left( -\frac{\epsilon^3 n}{24} \right).
    \end{align*}
    The second inequality holds because $\mathbb{E}\left[ \deg_S(v) \right]\ge \frac{\beta n}{1 + \epsilon} \ge \frac{\beta n}{2}$ by the sufficient degree constraint (\ref{lp:1_degrees_lowerbound}).
    Combining the above inequalities, we have,
    
    \begin{align}\label{eq:prob_bound_good_Si}
        &\Pr\left[\mathtt{OBJ}\cap \mathtt{EDGE} \bigcap_{v\in H} \mathtt{Deg}(v) \cap \mathtt{GoodSample}\right]\\
         =& \Pr\left[\left.\mathtt{OBJ}\cap \mathtt{EDGE} \bigcap_{v\in H} \mathtt{Deg}(v) \right| \mathtt{GoodSample}\right] \cdot \Pr[\mathtt{GoodSample}]\\
        \ge& \left(1 - \Pr[\mathtt{OBJ}^c\mid \mathtt{GoodSample}] - \Pr[\mathtt{EDGE}^c\mid \mathtt{GoodSample}] - \sum_{v\in H} \Pr[\mathtt{Deg}(v)^c \mid \mathtt{GoodSample}] \right) \cdot  \Pr[\mathtt{GoodSample}] \notag \\
        \ge& \left( 1 - \exp \left(  - 2 \epsilon^2\cdot \beta^2 \cdot n\right) - \exp\left( -\frac{2\epsilon^4 n}{9} \right)  - 2n\cdot\exp\left( -\frac{\epsilon^3 n}{24} \right) \right) \left( 1 - \frac{1}{2n^2} \right) \notag   \\
        \geq& 1 - \frac{1}{n^2},
    \end{align}
    where the second to last inequality holds due to  Observation~\ref{obsn:high-prob-event} and the last inequality holds because $n\geq \frac{2^2\cdot 3^4}{\epsilon^8}$ and $\epsilon \leq 1/7.$
\end{proof}

\end{document}